\documentclass[10pt,journal]{IEEEtran}
\usepackage{amssymb}
\usepackage{amsmath}
\usepackage{cite}
\usepackage{url}
\usepackage{color}
\usepackage{xcolor}
\usepackage{cite,graphicx,amsmath,amssymb}
\usepackage{subfigure}
\usepackage{fancyhdr}
\usepackage{comment}
\usepackage{mdwmath}
\usepackage{mdwtab}
\usepackage{tabularx}
\usepackage{caption}
\usepackage{amsthm}
\usepackage{setspace}
\usepackage{algorithm}
\usepackage{algorithmic}
\usepackage{pdfpages}
\usepackage{mathtools}
\usepackage[colorlinks=true, linkcolor=blue, citecolor=blue, urlcolor=blue]{hyperref}

\usepackage{bm}

\graphicspath{{./src}}

\newtheoremstyle{italicheaduprightbody}
  {\topsep}{\topsep}
  {\normalfont}{\parindent}
  {\normalfont\itshape}{:}{0.5em}
  {\thmname{#1}\thmnumber{ #2}\thmnote{\normalfont\ (#3)}}
\theoremstyle{italicheaduprightbody}

\newtheorem{remark}{Remark}
\newtheorem{theorem}{Theorem}

\newtheorem{lemma}{Lemma}

\newtheorem{corollary}{Corollary}

\newtheorem{proposition}{Proposition}

\allowdisplaybreaks

\title{\huge{Movable-Antenna-Empowered ISAC: 3D Ambiguity Function Analysis, Waveform Design and Antenna Position Optimization}}

\author{

Tianqi~Mao,~\IEEEmembership{Member,~IEEE},
Yuanshuo~Gang,~\IEEEmembership{Graduate Student Member,~IEEE},
Guangyao~Liu,~\IEEEmembership{Member,~IEEE},
Shiqi~Cui, 
Ying~Sun,
Yifeng~Xiong,~\IEEEmembership{Member,~IEEE},
Christos~Masouros,~\IEEEmembership{Fellow,~IEEE}, \\
and
Ana~García~Armada,~\IEEEmembership{Fellow,~IEEE}

\vspace{-0.8cm}

\thanks{The work of Tianqi Mao was supported in part by National Natural Science Foundation of China under Grants 62622105 and 62401054; and in part by the Beijing Nova Program under Grant 202604841146. The work of Ana García Armada was supported in part by Spanish National Project under Grant PID2023-147305OB-C31 (SOFIA-AIR) and project Grant HE101293106 (MULTIPLY-6G) funded by the European Union. (\emph{Corresponding authors: Guangyao Liu and Ying Sun.})}
\thanks{
    Tianqi Mao, Yuanshuo Gang, Guangyao Liu, Shiqi Cui, and Ying Sun are with the State Key Laboratory of Environment Characteristics and Effects for Near-space, Beijing Institute of Technology, Beijing 100081, China (e-mails: \{maotq, gangys, liu\_gy, cuisqqq, sunyingsy\}@bit.edu.cn).

    Yifeng Xiong is with the School of Information and Electronic Engineering, Beijing University of Posts and Telecommunications, Beijing 100876, China (e-mail: yifengxiong@bupt.edu.cn).
	
	Christos Masouros is with the Department of Electrical and Electronic Engineering, University College London, WC1E 7JE London, U.K. (e-mail: c.masouros@ucl.ac.uk).

    Ana García Armada is with the Department of Signal Theory and Communications, Universidad Carlos III de Madrid, 28911 Leganés, Spain (e-mail: anagar@ing.uc3m.es).
	}
}
\begin{document}

\maketitle

\begin{abstract}
Multiple-input multiple-output orthogonal frequency division multiplexing (MIMO-OFDM) provides abundant spatial and time-frequency degrees of freedom for integrated sensing and communication (ISAC), while movable-antennas (MAs) further introduce reconfigurable spatial freedom through array geometry adjustment. 
However, how the array geometry and the MIMO-OFDM waveform design jointly influence the sensing performance, i.e., the three dimensional (3D) ambiguity function property, remains insufficiently understood. 
Against this background, we investigate MA-empowered ISAC design under the MIMO-OFDM architecture, where the radio-frequency (RF) antenna positions and the baseband waveform are co-designed with symbol-level precoding (SLP) philosophy. 
The discrete periodic angle-range-Doppler ambiguity function is firstly derived for the MA-enabled ISAC framework, revealing its interesting features with respect to the waveform matrix rank. 
Specifically, it is proved that a rank-one waveform yields a 3D ambiguity function that is separable into angular and range-Doppler components.
On the other hand, a higher-rank waveform combines multiple spatial components with angle-dependent weights, allowing the range-Doppler ambiguity function to vary across candidate angles and enabling the MA positions to further shape the 3D ambiguity function.
Based on the aforementioned analysis, a 3D integrated sidelobe level ratio (ISLR) is defined, a joint optimization problem of the MA positions and the SLP waveform matrix is formulated, and a penalty dual decomposition based alternating optimization algorithm is developed. 
For the radar-only case, the existence of a rank-one globally optimal waveform is established, which reduces the original high-dimensional waveform optimization to a low-dimensional joint design of the spatial beam and MA positions.
Numerical results verify that the proposed approach can significantly suppress the 3D ambiguity sidelobes and improve weak target detection while maintaining reliable multiuser communication performance.
\end{abstract}


%\vspace{+0.3cm}
\section{Introduction} \label{sec:introduction}
\IEEEPARstart{T}{he} ongoing development of sixth-generation (6G) communication networks is envisioned to support a plethora of emerging applications such as the low-altitude economy, intelligent manufacturing, and autonomous driving which facilitate increasing demands on both high-rate reliable communications and accurate environmental sensing. 
The integrated sensing and communication (ISAC) philosophy can address both requirements by fusing information transmission and environmental sensing over shared spectrum and hardware resources, thereby attaining efficient resource utilization together with potential synergistic gains \cite{zhang2026evolution,mao2022waveform,mendez2026_superimposed,zhang2026afdm,chen2026_privacy}. 
Therefore, ISAC has been widely recognized as an indispensable enabling technology for future 6G networks and beyond.

Owing to its attractive advantages for both communication and sensing and extensive deployment in existing wireless systems, orthogonal frequency division multiplexing (OFDM) has become one competitive waveform candidate for ISAC \cite{zhang2024cross_domain,liu2025cp_ofdm}. 
In particular, its multicarrier time-frequency structure supports spectrally efficient communications and range-Doppler sensing \cite{zhang2025target_localization,sturm2011waveform}.
%For communication, OFDM enables efficient parallel transmission over orthogonal subcarriers and mitigates frequency selective fading through a cyclic prefix (CP) and frequency-domain equalization, thereby providing high spectral efficiency with relatively low implementation complexity \cite{zhang2025target_localization}. 
%For sensing, the multicarrier time-frequency structure of OFDM provides natural delay and Doppler processing dimensions through the phase variations across subcarriers and OFDM symbols, respectively, enabling range and velocity estimation \cite{sturm2011waveform}. 
Moreover, multiple-input multiple-output (MIMO) technology introduces additional spatial degrees of freedom (DoF), enhancing transmission performance while providing richer spatial information for target sensing across the angular-domain \cite{dai2026tutorial,heath2026trihybrid}. 
Consequently, MIMO-OFDM can simultaneously provide significant time-frequency-space DoF for both communication and sensing system design, making it an important waveform framework for ISAC.

\vspace{-0.2cm}
\subsection{Related Works}\label{subsec:related_works}
Most of the MIMO-OFDM ISAC studies focused on dual-functional waveform optimization and beamforming. 
Specifically, existing approaches mainly employed block-level precoding (BLP) for beamformer design or optimized transmit covariance matrices. The considered metrics included achievable rate \cite{gang2025uav}, user signal-to-interference-plus-noise ratio (SINR) \cite{liu2018mu_mimo}, multiuser interference \cite{liu2018dual_functional}, and sensing metrics including sensing SINR \cite{hatami2026beamforming}, beampattern mismatch \cite{liu2025beam_pattern}, and the Cramér–Rao bound (CRB) \cite{zhang2026leo}, etc.
Besides, the ambiguity function was also employed in BLP design to suppress range-angle sidelobes \cite{liu2025range_angle}. 
Despite these advances, conventional BLP typically optimizes a common precoder based on statistical metrics, limiting control over how random communication symbols shape the transmitted waveform and its ambiguity sidelobes.

Symbol-level precoding (SLP) offers a promising approach to addressing this limitation. 
By exploiting instantaneous communication symbols, SLP directly optimizes the transmit waveform while converting part of the multiuser interference into constructive interference (CI) \cite{masouros2015green,li2020tutorial}. 
Accordingly, SLP has been applied to ISAC for beampattern shaping \cite{liu2021dual_functional_slp,nguyen2026exploiting}, target illumination \cite{wang2025symbol_scaling}, and parameter estimation \cite{liao2024ftn}. 
For MIMO-OFDM ISAC, SLP can further exploit the DoF across individual resource elements (REs) to directly shape the time-frequency waveform structure, thereby enabling range-Doppler ambiguity sidelobe suppression \cite{li2025mimo_ofdm,li2025low_complexity,cai2026distributed}. 
Specifically, the discrete periodic range-Doppler ambiguity function of MIMO-OFDM signals was derived in \cite{li2025mimo_ofdm}, where the SLP waveform on each RE was jointly optimized to minimize the integrated sidelobe level (ISL). 
Building on this framework, a low-complexity waveform design method was developed in \cite{li2025low_complexity}. 
The work in \cite{cai2026distributed} extended SLP from single-cell designs to coordinated multi-cell MIMO-OFDM ISAC, considering intercell interference and target illumination.
%However, these designs all assume fixed position arrays (FPAs), whose array manifolds are determined by preset antenna locations and cannot be reconfigured according to sensing tasks.

%Array geometry provides another important design DoF for sensing. 
%Sparse array and antenna selection studies \cite{wei2024joint,xie2024cross} have shown that antenna placement significantly affects spatial resolution and ambiguity characteristics, but conventional array geometries remain fixed after deployment and cannot adapt to changing sensing environments or tasks. 
%Movable-antennas (MAs) address this limitation by allowing antenna positions to be reconfigured within prescribed regions \cite{xiao2026_fundamental,zhu2025_movable,zhu2026_multiuser}. 
To further enhance sensing performance beyond the fixed position arrays (FPAs) adopted in the aforementioned SLP designs, movable antennas (MAs) allow the array geometry to be reconfigured within prescribed regions \cite{xiao2026_fundamental,zhu2025_movable,zhu2026_multiuser}. 
Leveraging this flexibility, the ambiguity function of MA-enabled frequency hopping MIMO radar was derived in \cite{chen2025fh_mimo}, where the MA positions were optimized to balance mainlobe width and sidelobe levels. 
In \cite{ma2026robust_sensing}, the tradeoff between the CRB and ambiguity sidelobes across different SNR regimes was characterized, and the MA positions were optimized to balance estimation accuracy and ambiguity suppression. 
In addition, the work \cite{wu2026fluid_radar} jointly designed the transmit and receive MA positions, transmit waveform, and receive filter to improve target detection in interference-rich environments. 
These studies, however, focus on radar-only sensing and do not account for the symbol dependent waveform constraints imposed by random communication data. 
Existing ISAC designs with MAs mainly optimize spatial metrics \cite{khalili2026movable,li2026uav_isac}, with limited consideration of waveform correlation across the delay and Doppler dimensions. 
More recently, the study \cite{feng2026range_doppler} jointly optimized the MA positions and transmit beamformers for range-Doppler sidelobe suppression. 
%Nevertheless, this method still employs conventional BLP, in which the precoder on each subcarrier is shared across OFDM symbols, and evaluates only the two dimensional (2D) range-Doppler ambiguity response at the target direction.
Nevertheless, this method employed conventional BLP, with each subcarrier precoder shared across OFDM symbols, and minimized the two dimensional (2D) range-Doppler ISL only at the target direction.

\subsection{Motivation and Contributions}\label{subsec:motivation_contributions}
%The above review reveals that MA-empowered MIMO-OFDM ISAC remains insufficiently explored. Existing studies either optimize waveforms on individual REs for FPAs or exploit MAs to improve spatial responses and 2D range-Doppler sidelobes. Consequently, the coupling between reconfigurable array geometry and SLP waveforms on individual REs remains unclear. When both are jointly designable, the array geometry changes the spatial projection of the transmitted waveform across candidate angles, which can make the resulting range-Doppler response vary with the candidate angle. Therefore, evaluating only the 2D range-Doppler response at the target direction cannot fully characterize the 3D ambiguity response jointly determined by the array geometry and transmit waveform.
The above review reveals that MA-empowered MIMO-OFDM ISAC remains insufficiently explored. 
As indicated by the aforementioned MA-ISAC studies, existing designs mainly focus on antenna positioning and spatial beamforming, with limited attention to joint waveform and beamforming design. 
Meanwhile, existing MIMO-OFDM ISAC studies have largely relied on conventional BLP or employed SLP for 2D range-Doppler sidelobe suppression at a known target direction, without explicitly considering sidelobes across the full angle-range-Doppler domain.

Motivated by this observation, we investigate the joint design of the transmit MA positions and SLP waveform for MA-empowered MIMO-OFDM ISAC, with emphasis on characterizing their joint effects on the angle-range-Doppler ambiguity function\footnote{For brevity, the angle-range-Doppler ambiguity function is hereafter referred to as the three dimensional (3D) ambiguity function.} and suppressing its sidelobes. The main contributions are summarized as follows.

\begin{itemize}
    \item We propose an MA-empowered MIMO-OFDM ISAC framework that combines reconfigurable transmit array geometry with SLP waveform design on each time-frequency RE. The transmit MAs provide additional spatial design freedom by adjusting the array geometry, while SLP directly designs the multi-antenna transmit waveform according to the instantaneous communication symbols, thereby jointly exploiting the array geometry and the spatial and time-frequency DoF of MIMO-OFDM.

    \item We first derive the discrete periodic 3D ambiguity function of the MA-empowered MIMO-OFDM waveform and characterize its structure from the perspective of waveform rank. We prove that, for a rank-one waveform, the 3D ambiguity function is \textit{separable} into angular and range-Doppler components, with all candidate angles sharing the same normalized range-Doppler profile. Conversely, under certain conditions, such separability implies a rank-one waveform. For a higher-rank waveform, multiple spatial basis components are combined with angle-dependent weights, allowing the range-Doppler profiles to vary across candidate angles and enabling the MA positions to further shape the 3D ambiguity sidelobes.

    \item Based on the derived 3D ambiguity function, we define a 3D integrated sidelobe level ratio (ISLR) over the angle-range-Doppler domain and formulate a joint optimization problem for the transmit MA positions and SLP waveform. To address the resulting strongly coupled nonconvex problem, we develop a penalty dual decomposition based alternating optimization (PDD-AO) algorithm. For the radar-only case, we further prove that at least one rank-one globally optimal solution exists, yielding an equivalent low-dimensional formulation involving a spatial beam and the MA positions.

\end{itemize}

Numerical results validate the effectiveness of the proposed design. Compared with FPA SLP-OFDM and conventional zero forcing (ZF)-OFDM benchmarks, the proposed scheme achieves lower 3D ISLR and substantially mitigates the masking of weak target by sidelobe leakage from strong target. It also maintains consistent sensing gains under different communication quality of service (QoS) requirements, numbers of users, and numbers of transmit antennas, demonstrating the benefit of jointly designing the array geometry and waveform on each RE.

The remainder of this paper is organized as follows. Section~\ref{section:system_model} introduces the MA-empowered MIMO-OFDM ISAC system model. Section~\ref{section:3D_AF_analysis} derives the 3D ambiguity function and analyzes its separability and dependence on waveform rank and MA positions. Section~\ref{section:joint_design} presents the joint waveform and antenna position optimization, develops the PDD-AO algorithm, and derives an equivalent low-dimensional formulation for the radar-only case. Numerical results are provided in Section~\ref{section:simulation}, followed by concluding remarks in Section~\ref{section:conclusion}.

\textit{Notations:} Scalars, vectors, matrices, and sets are denoted by italic, bold lowercase, bold uppercase, and calligraphic letters, respectively. Superscripts $(\cdot)^T$, $(\cdot)^H$, and $(\cdot)^*$ denote transpose, conjugate transpose, and complex conjugate. The operators $\Re\{\cdot\}$, $\Im\{\cdot\}$, and $\|\cdot\|_{2,F,\infty}$ denote the real part, imaginary part, and the corresponding norms. Moreover, $\operatorname{vec}(\cdot)$, $\operatorname{col}(\cdot)$, and $\operatorname{rank}(\cdot)$ denote vectorization, column space, and matrix rank, while $\operatorname{diag}(\mathbf a)$ forms a diagonal matrix from $\mathbf a$. $\mathbb R$, $\mathbb C$, $\mathbb Z$, and $\mathbf I_N$ denote the real, complex, and integer sets and the $N\times N$ identity matrix, respectively.

\section{System Model} \label{section:system_model}

\begin{figure}[t]
\centering
\includegraphics[width=\linewidth]{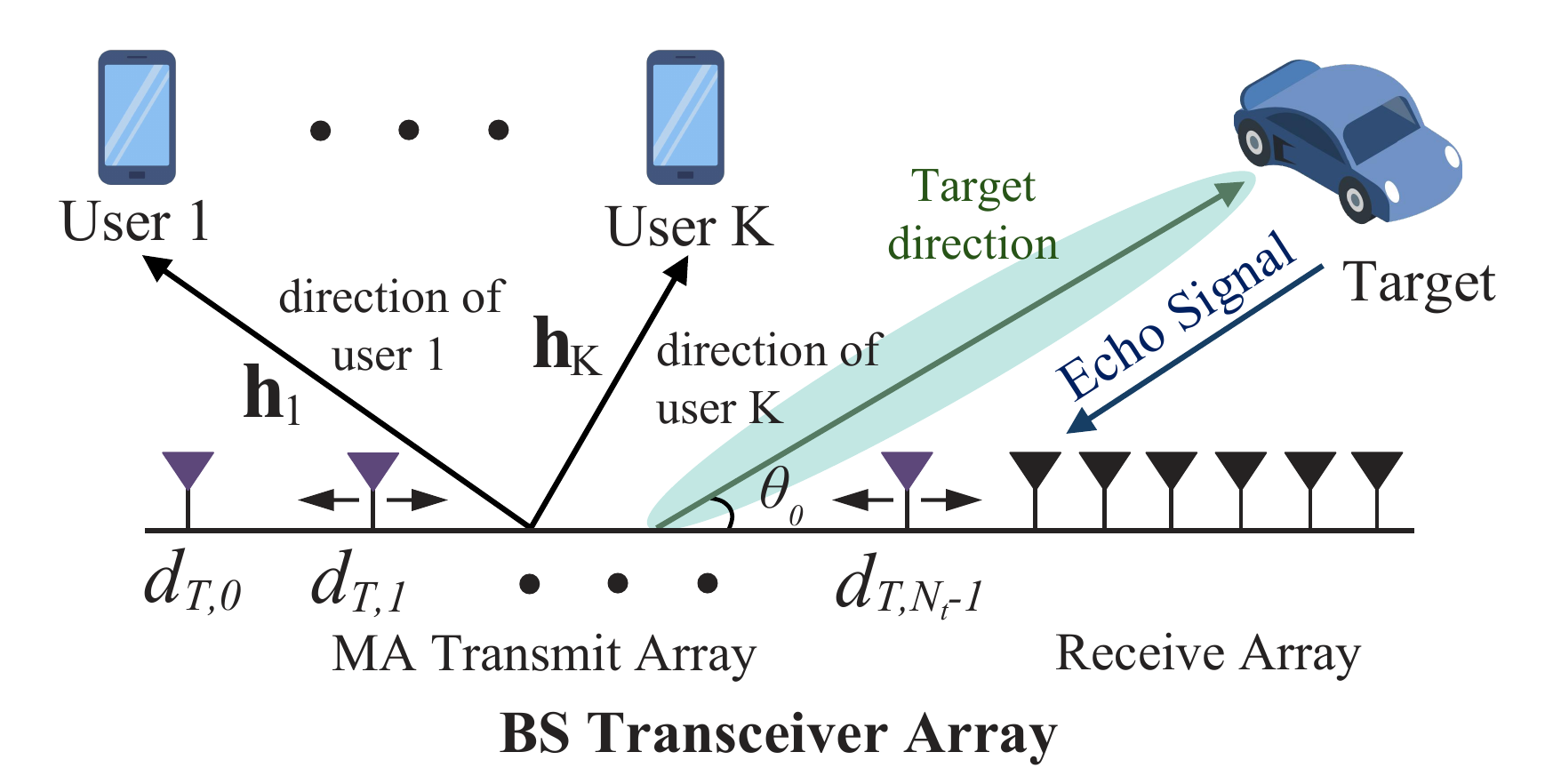}
\caption{Illustration of the MA-empowered MIMO-OFDM ISAC system.}
\label{fig:system_model}
\end{figure}

We consider a monostatic MA-empowered MIMO-OFDM ISAC system shown in Fig.~\ref{fig:system_model}, where the base station (BS) equipped with $N_t$ transmit MAs and a fixed $N_r$-element receive uniform linear array (ULA) simultaneously serves $K$ single-antenna users. Besides, one single target is assumed within the far-field region for simplicity. The BS transmits a dual-functional MIMO-OFDM waveform over $N_c$ subcarriers and $N_s$ consecutive OFDM symbols, constituting one OFDM frame. The transmit MAs remain stationary within each frame and can be reconfigured across different OFDM frames \cite{ma2026robust_sensing}. Their positions within an one dimensional interval $[0, D_T]$ are collected in $\mathbf d_T=[d_{T,0},\ldots,d_{T,N_t-1}]^T\in\mathbb R^{N_t}$. The feasible MA position set is $\mathcal D_T \triangleq \{\mathbf d_T\mid d_{T,0}=0,\, d_{T,N_t-1}\le D_T,\, d_{T,i+1}-d_{T,i}\ge d_{\min},i=0,\ldots,N_t-2\}$, where $d_{T,0}=0$ fixes the first MA as the position reference, and $d_{\min}$ denotes the minimum spacing between adjacent MAs for limiting mutual coupling.

Under the far-field condition, the transmit array steering vector associated with azimuth angle $\theta$ is given by\footnote{Note that $\theta$ is a generic angular variable applicable to both sensing and communication propagation directions.}\vspace{-0.1cm}
\begin{equation}
\mathbf a_T(\theta;\mathbf d_T)
=
\begin{bmatrix}
1,
e^{j k_c d_{T,1}\sin\theta},
\ldots,
e^{j k_c d_{T,N_t-1}\sin\theta}
\end{bmatrix}^{T},
\label{eq:tx_array_steering_vector}
\end{equation}
where $k_c=2\pi/\lambda=2\pi f_c/c_0$, with $f_c$, $c_0$, and $\lambda=c_0/f_c$ denoting the carrier frequency, the speed of light, and the wavelength, respectively. Under the monostatic configuration, the transmit and receive arrays observe the same target azimuth angle \cite{dai2026tutorial}. Hence, by assuming the antenna spacing $d_R=\lambda/2$, the receive array steering vector is expressed as\vspace{-0.1cm}
\begin{equation}
\mathbf a_R(\theta)
=
\begin{bmatrix}
1,
e^{j k_c d_R\sin\theta},
\ldots,
e^{j k_c (N_r-1)d_R\sin\theta}
\end{bmatrix}^{T}.
\label{eq:rx_array_steering_vector}
\end{equation}

\subsection{Transmit Signal Model}\label{subsec:transmit_model}
On the $(m,n)$-th RE, the information symbol vector for the $K$ communication users is denoted by $\mathbf s_{n,m}=[s_{n,m,1},\ldots,s_{n,m,K}]^T\in\mathbb C^K$. Without loss of generality, we consider $M$-ary phase shift keying (PSK) modulation with unit modulus\footnote{$M$-PSK is adopted for a compact CI formulation and analytical tractability. The proposed framework can be extended to $M$-quadrature amplitude modulation (QAM) by replacing the PSK CI constraints with the corresponding QAM decision-region constraints for inner and outer constellation points \cite{liu2021dual_functional_slp,li2025mimo_ofdm}.}, i.e., $s_{n,m,k}\in\mathcal S_M\triangleq \{e^{j2\pi \mu/M}\mid \mu=0,\ldots,M-1\}$. By applying SLP operations, $\mathbf s_{n,m}$ is directly mapped to the multi-antenna transmit waveform $\mathbf x_{n,m}=[x_{0,n,m},\ldots,x_{N_t-1,n,m}]^T\in\mathbb C^{N_t}$ on each RE, i.e., $\mathbf s_{n,m}\mapsto\mathbf x_{n,m}$.

By collecting the transmit waveforms over all REs, the complex baseband time-domain transmit signal $\widetilde{\mathbf x}(t)\in\mathbb C^{N_t}$ is given by\vspace{-0.1cm}
\begin{equation}
\widetilde{\mathbf x}(t)=\frac{1}{\sqrt{N_c}}\sum_{m=0}^{N_s-1}\sum_{n=0}^{N_c-1}\mathbf x_{n,m}e^{j2\pi n\Delta f t_m}\operatorname{rect}\left(\frac{t_m-\frac{T-T_{\rm cp}}{2}
}{T_0}\right),
\label{eq:tx_baseband_signal}
\end{equation}
%where $\Delta f=1/T$ is the subcarrier spacing, $T$ and $T_{\rm cp}$ represent the useful OFDM symbol and cyclic prefix (CP) durations, respectively, and $T_0=T+T_{\rm cp}$, while $t_m=t-mT_0$ is referenced to the beginning of the useful part of the $m$-th OFDM symbol. The rectangular window $\operatorname{rect}(u)$ equals one for $|u|\le1/2$ and zero otherwise, so that the $m$-th OFDM symbol including its CP occupies $-T_{\rm cp}\le t_m<T$. After upconversion, the transmitted radio frequency signal can be written as $\widetilde{\mathbf x}_{\rm RF}(t)=\Re\left\{\widetilde{\mathbf x}(t)e^{j2\pi f_ct}\right\}$.
where $\Delta f=1/T$ is the subcarrier spacing, $T$ and $T_{\mathrm{cp}}$ denote the useful OFDM symbol and cyclic prefix (CP) durations, respectively, and $T_0=T+T_{\mathrm{cp}}$. The local time variable $t_m=t-mT_0$ is measured from the beginning of the useful part of the $m$-th OFDM symbol, such that the CP and useful part occupy $-T_{\mathrm{cp}}\le t_m<0$ and $0\le t_m<T$, respectively. The rectangular window $\operatorname{rect}(u)$ equals one for $-1/2\le u<1/2$ and zero otherwise. After up conversion, the transmitted radio frequency signal can be written as $\widetilde{\mathbf x}_{\rm RF}(t)=\Re\left\{\widetilde{\mathbf x}(t)e^{j2\pi f_ct}\right\}$.

\subsection{Radar Echo Signal Model}\label{subsec:radar_model}
Consider a point target in the far-field with true parameters $\boldsymbol{\eta}_0\triangleq(\theta_0,\tau_0,f_{D,0})$, where $\theta_0$ is the target azimuth angle, while $\tau_0=2R_0/c_0$ and $f_{D,0}=2v_0f_c/c_0$ denote the round trip propagation delay and Doppler frequency shift induced by the target range $R_0$ and radial velocity $v_0$, respectively.
The target parameters and complex reflection coefficient are assumed to remain constant within the current OFDM frame. 
Besides, let $B=N_c\Delta f$ and $D_R=(N_r-1)d_R$ denote the signal bandwidth and receive array aperture, respectively.
We adopt a spatial narrowband model and assume $BD_T/c_0\ll1$ and $BD_R/c_0\ll1$.
%Accordingly, the frequency dependence of the array manifolds across subcarriers can be neglected, and all subcarriers share the transmit and receive steering vectors evaluated at the carrier frequency $f_c$.

After target reflection and down conversion, the baseband time-domain echo signal
$\widetilde{\mathbf y}(t)\in\mathbb C^{N_r}$ received at the BS can be expressed as\vspace{-0.1cm}
\begin{equation}
\widetilde{\mathbf y}(t)
=
\bar{\alpha}_0
\mathbf a_R(\theta_0)
\mathbf a_T^H(\theta_0;\mathbf d_T)
\widetilde{\mathbf x}(t-\tau_0)
e^{j2\pi f_{D,0}t}
+
\widetilde{\mathbf n}(t),
\label{eq:radar_echo_time}
\end{equation}
where $\bar{\alpha}_0=\alpha_0e^{-j2\pi f_c\tau_0}$ and $\alpha_0$ is the complex reflection coefficient incorporating propagation attenuation and target scattering characteristics, while $\widetilde{\mathbf n}(t)\in\mathbb C^{N_r}$ denotes the additive white Gaussian noise (AWGN) with zero mean and variance $\sigma_r^2$.

We further assume $0\le\tau_0<T_{\rm cp}$, such that no interference between OFDM symbols remains after CP removal.
Besides, we assume $|f_{D,0}|/\Delta f\ll1$, so that the Doppler phase variation within each OFDM symbol and the resulting inter-carrier interference (ICI) can be neglected \cite{sturm2011waveform}.
After sampling, CP removal, and an $N_c$-point discrete Fourier transform (DFT) with unitary normalization, the frequency-domain echo signal
$\mathbf y_{n,m}\in\mathbb C^{N_r}$ on the $(n,m)$-th RE can be written as\vspace{-0.1cm}
\begin{equation}
\label{eq:radar_echo_freq}
\begin{aligned}
\mathbf y_{n,m}
=&
\bar{\alpha}_0
\mathbf a_R(\theta_0)
\mathbf a_T^H(\theta_0;\mathbf d_T)
\mathbf x_{n,m}
e^{j2\pi(f_{D,0}mT_0-n\Delta f\tau_0)} \\[-1mm]
&+
\mathbf n_{n,m}.
\end{aligned}
\end{equation}
where $\mathbf n_{n,m}\sim\mathcal{CN}(\mathbf 0,\sigma_r^2\mathbf I_{N_r})$ denotes the frequency-domain AWGN.

\subsection{Communication Model}\label{subsec:comm_model}
For user $k$, variations in the MA positions change the spatial phases of the propagation paths across different transmit antennas. Therefore, the communication channel explicitly depends on the MA position vector $\mathbf d_T$. For a frequency-selective channel with $L_k$ effective propagation paths, the frequency-domain channel $\mathbf h_{n,k}(\mathbf d_T)\in\mathbb C^{N_t}$ between the BS and user $k$ on the $n$-th subcarrier is expressed as
\begin{equation}
\mathbf h_{n,k}(\mathbf d_T)
=
\sum_{\ell=1}^{L_k}
\beta_{k,\ell}
e^{-j2\pi n\Delta f\tau_{k,\ell}}
\mathbf a_T
\left(
\varphi_{k,\ell};\mathbf d_T
\right),
\label{eq:comm_channel}
\end{equation}
where $\beta_{k,\ell}$, $\tau_{k,\ell}$, and $\varphi_{k,\ell}$ denote the complex gain, propagation delay, and angle of departure of the $\ell$-th propagation path, respectively. Since the feasible MA movement region is much smaller than the propagation distance between the BS and the user, these path parameters are approximately invariant over the feasible region. We further assume that the BS has the position-dependent channel state information required for the joint design \cite{cheng2026integrated}. Accordingly, the frequency-domain signal received by user $k$ on the $(n,m)$-th RE is given by\vspace{-0.1cm}
\begin{equation}
y_{n,m,k}
=
\mathbf h_{n,k}^{H}(\mathbf d_T)
\mathbf x_{n,m}
+
z_{n,m,k},
\label{eq:comm_received_signal}
\end{equation}
where $z_{n,m,k}\sim\mathcal{CN}(0,\sigma_{c,k}^{2})$ represents the AWGN at user $k$. Based on the SLP model in Section~\ref{subsec:transmit_model}, the communication QoS of user $k$ on the $(n,m)$-th RE is enforced through the following $M$-PSK CI constraint:\vspace{-0.1cm}
\begin{equation}
\left[
\Re\{\zeta_{n,m,k}\}
-
\sigma_{c,k}\sqrt{\Gamma_k}
\right]
\tan\phi_M
\ge
\left|
\Im\{\zeta_{n,m,k}\}
\right|,
\label{eq:ci_constraint}
\end{equation}
where $\zeta_{n,m,k} = \mathbf h_{n,k}^{H}(\mathbf d_T) \mathbf x_{n,m}s_{n,m,k}^{*}$ denotes the rotated noiseless received symbol, $\phi_M=\pi/M$ is the half angle of the $M$-PSK decision region, and $\Gamma_k$ is the communication QoS threshold of user $k$.

To reformulate the CI constraint in \eqref{eq:ci_constraint} as two affine inequalities, we define\vspace{-0.1cm}
\begin{equation}
\begin{aligned}
&\widetilde{\mathbf h}_{n,m,k'}^{H}(\mathbf d_T)\\
&=
\begin{cases}
\mathbf h_{n,k}^{H}(\mathbf d_T)
s_{n,m,k}^{*}
\left(
\sin\phi_M+j\cos\phi_M
\right),
&
k'=2k-1,
\\[1mm]
\mathbf h_{n,k}^{H}(\mathbf d_T)
s_{n,m,k}^{*}
\left(
\sin\phi_M-j\cos\phi_M
\right),
&
k'=2k ,
\end{cases}
\end{aligned}
\label{eq:extended_comm_channel}
\end{equation}
with the corresponding CI thresholds given by\vspace{-0.1cm}
\begin{equation}
\gamma_{2k}
=
\gamma_{2k-1}
=
\sigma_{c,k}\sqrt{\Gamma_k}\sin\phi_M,
\quad k=1,\ldots,K.
\label{eq:ci_threshold}
\end{equation}
Then, \eqref{eq:ci_constraint} can be equivalently written as\vspace{-0.1cm}
\begin{equation}
\Re
\left\{
\widetilde{\mathbf h}_{n,m,k'}^{H}
(\mathbf d_T)
\mathbf x_{n,m}
\right\}
\ge
\gamma_{k'},
\quad
k'=1,\ldots,2K.
\label{eq:affine_ci_constraint}
\end{equation}

\section{3D Ambiguity Function Analysis}\label{section:3D_AF_analysis}
This section investigates how the transmit MA geometry and MIMO-OFDM waveform jointly shape the 3D ambiguity function. We first derive the discrete periodic 3D ambiguity function and examine when waveform phase compensation can offset MA position changes. Then, we relate waveform rank to the separability between the angular dimension and the range-Doppler plane and characterize the impact of MA positions under rank-one and higher-rank waveforms.

\subsection{Discrete Periodic 3D Ambiguity Function}
\label{subsec:discrete_periodic_3d_af}
Let $(\theta_0,\tau_0,f_{D,0})$ denote the true target parameters and $(\theta,\tau,f_D)$ denote the candidate parameters used for matched processing. The MIMO radar ambiguity function characterizes the matched response of the transmit waveform and array under mismatched candidate target parameters \cite{sanantonio2007mimo_radar}. Accordingly, the continuous 3D ambiguity function can be calculated as\vspace{-0.1cm}
%The MIMO radar ambiguity function characterizes the matched response of the transmit waveform and array under mismatched candidate target parameters; its mainlobe and sidelobes reflect the parameter-resolution capability and global ambiguity characteristics, respectively \cite{sanantonio2007mimo_radar}. Accordingly, the continuous 3D ambiguity function is defined as
\begin{equation}
\label{eq:continuous_3d_af}
\begin{aligned}
\chi_{\theta_0}(\theta,\Delta\tau,\Delta f_D)
\triangleq&
\rho_R(\theta,\theta_0)
\int_{-\infty}^{+\infty}
\mathbf a_T^H(\theta_0;\mathbf d_T)
\widetilde{\mathbf x}(t)
\\[-1mm]
&\times
\widetilde{\mathbf x}^{H}(t+\Delta\tau)
\mathbf a_T(\theta;\mathbf d_T)
e^{j2\pi\Delta f_Dt}\,
\mathrm{d}t,
\end{aligned}
\end{equation}
where $\rho_R(\theta,\theta_0)=\mathbf a_R^H(\theta)\mathbf a_R(\theta_0)$, $\Delta\tau=\tau_0-\tau$, and $\Delta f_D=f_{D,0}-f_D$. 
%It can be seen that \eqref{eq:continuous_3d_af} retains both the transmit--receive spatial mismatch and the delay--Doppler mismatch, thereby characterizing the 3D matched response jointly determined by the array geometry and transmit waveform.\footnote{In particular, when the candidate angle coincides with the true target angle, \eqref{eq:continuous_3d_af} reduces to the two-dimensional delay--Doppler ambiguity function formed along $\theta_0$. For a single-antenna system, it further reduces to the classical Woodward ambiguity function.}

Since $\tau_0<T_{\rm cp}$, the delay after CP removal can be represented by a circular shift within the useful OFDM symbol. The $q$-th time-domain transmit sample of the $m$-th OFDM symbol after CP removal, denoted by $\widetilde{\mathbf x}_{q,m}\in\mathbb C^{N_t}$, is written as\vspace{-0.1cm}
\begin{equation}
\label{eq:time_domain_ofdm_sample}
\widetilde{\mathbf x}_{q,m}
=
\frac{1}{\sqrt{N_c}}
\sum_{n=0}^{N_c-1}
\mathbf x_{n,m}
e^{j2\pi qn/N_c},
\quad
q=0,\ldots,N_c-1.
\end{equation}
Substituting \eqref{eq:time_domain_ofdm_sample} into the periodic correlation expression and using Fourier orthogonality, the discrete periodic 3D ambiguity function can be derived as\vspace{-0.1cm}
\begin{equation}
\label{eq:discrete_periodic_3d_af}
\begin{aligned}
\chi_{\theta_0}(\theta,l,\nu)
\triangleq&
\rho_R(\theta,\theta_0)
\sum_{m=0}^{N_s-1}
\sum_{q=0}^{N_c-1}
\mathbf a_T^H(\theta_0;\mathbf d_T)
\widetilde{\mathbf x}_{q,m}
\\[-1mm]
&\times
\widetilde{\mathbf x}_{[q+l]_{N_c},m}^{H}
\mathbf a_T(\theta;\mathbf d_T)
e^{j2\pi\nu m/N_s}
\\
=&
\rho_R(\theta,\theta_0)
\sum_{m=0}^{N_s-1}
\sum_{n=0}^{N_c-1}
\mathbf a_T^H(\theta_0;\mathbf d_T)
\mathbf x_{n,m}\mathbf x_{n,m}^{H}
\\[-1mm]
&\times
\mathbf a_T(\theta;\mathbf d_T)
e^{-j2\pi ln/N_c}
e^{j2\pi\nu m/N_s},
\end{aligned}
\end{equation}
where $[q+l]_{N_c}$ represents the cyclic shift modulo $N_c$, while $l=0,\ldots,N_c-1$ and $\nu=0,\ldots,N_s-1$ denote the discrete range and Doppler bins over one complete period, corresponding to $\Delta\tau=l/(N_c\Delta f)$ and $\Delta f_D=\nu/(N_sT_0)$, respectively. Note that, unlike the 2D range-Doppler ambiguity functions considered at the target direction in \cite{li2025mimo_ofdm,feng2026range_doppler}, \eqref{eq:discrete_periodic_3d_af} retains the candidate angle $\theta$. Hence, the spatial outer product of each RE is projected onto the true and candidate array responses before the range and Doppler Fourier weighting, allowing the MA positions and transmit waveform to jointly shape the ambiguity function across candidate angles.

%Thus, full-period correlation completely eliminates the cross-subcarrier terms with $n\neq n'$ and retains only the autocorrelation contribution of each RE. Compared with MIMO-OFDM ambiguity functions that examine only the two-dimensional range-Doppler response in the target direction \cite{li2025mimo_ofdm,feng2026range_doppler}, \eqref{eq:discrete_periodic_3d_af} further retains the candidate angle $\theta$, thereby explicitly characterizing the ambiguity response under angular mismatch. Specifically, the spatial outer product $\mathbf x_{n,m}\mathbf x_{n,m}^{H}$ at each RE is weighted by the range-Doppler phase terms and then bilaterally projected onto the array manifolds corresponding to the true target direction and the candidate direction. Consequently, the MA geometry and transmit waveform jointly determine the 3D ambiguity structure at different candidate angles.

To clarify the role of the candidate angle, consider two feasible MA position vectors $\mathbf d_T$ and $\mathbf d_T'$, and let $\Delta d_i=d_{T,i}'-d_{T,i}$, $i=0,\ldots,N_t-1$. Since the first MA is the position reference, $\Delta d_0=0$. Define $\mathbf D_0\triangleq\operatorname{diag}\{e^{jk_c\Delta d_i\sin\theta_0}\}_{i=0}^{N_t-1}$, and let $\mathbf x_{n,m}'=\mathbf D_0\mathbf x_{n,m}$. Then, $\mathbf a_T^H(\theta_0;\mathbf d_T')\mathbf x_{n,m}'=\mathbf a_T^H(\theta_0;\mathbf d_T)\mathbf x_{n,m}$, so the entire range and Doppler slice at $\theta=\theta_0$ is preserved. Whether the same equivalence holds at other candidate angles is characterized below.

\begin{proposition}
\label{prop:ma_phase_compensation}
For the above MA position change and phase compensation toward $\theta_0$, the transmit response at any candidate angle $\theta$ satisfies\vspace{-0.1cm}
\begin{equation}
\label{eq:phase_compensated_response}
\mathbf a_T^H(\theta;\mathbf d_T')
\mathbf x_{n,m}'=
\mathbf a_T^H(\theta;\mathbf d_T)
\mathbf D_{\Delta}(\theta,\theta_0)
\mathbf x_{n,m},
\end{equation}
where
$\mathbf D_{\Delta}(\theta,\theta_0)\triangleq
\operatorname{diag}\{e^{-jk_c\Delta d_i(\sin\theta-\sin\theta_0)}\}_{i=0}^{N_t-1}$.
For $\theta\neq\theta_0$, invariance of the candidate angle response for every $\mathbf x_{n,m}\in\mathbb C^{N_t}$ holds if and only if\vspace{-0.1cm}
\begin{equation}
\label{eq:angle_invariance_condition}
k_c\Delta d_i(\sin\theta-\sin\theta_0)
\in 2\pi\mathbb Z,
\quad i=0,\ldots,N_t-1.
\end{equation}
Equivalently, $\mathbf D_{\Delta}(\theta,\theta_0)=\mathbf I_{N_t}$.
Since $\Delta d_0=0$, the same condition is also necessary and sufficient for invariance of the response magnitude for every transmit vector.
\end{proposition}
\begin{proof}
See \textbf{App.} A in the supplemental document \cite{mao2026detailed}.
%Please refer to Appendix~\hyperref[app:appendix_a]{A}.
\end{proof}

\begin{remark}
\label{remark:ma_angle_selectivity}
Phase compensation that preserves the range and Doppler slice at $\theta_0$ generally fails for $\theta\neq\theta_0$. Thus, different MA positions and transmit waveforms can yield the same response at the target direction but different range-Doppler profiles at other candidate angles.
\end{remark}

%\begin{remark}
%\label{remark:ma_angle_selectivity}
%Proposition~\ref{prop:ma_phase_compensation} shows that per-antenna phase compensation for the true target angle $\theta_0$ can completely cancel the spatial phase differences caused by MA position variations in that direction, thereby preserving the entire range-Doppler slice. However, for $\theta\neq\theta_0$, the equivalence established by the array-position variation and waveform phase compensation generally no longer holds. Therefore, even when different combinations of MA positions and transmit waveforms yield identical range-Doppler responses in the target direction, they may still produce different ambiguity responses at other candidate angles.
%\end{remark}

\subsection{Ambiguity Function Separability Between the Angular Dimension and Range-Doppler Plane}
\label{subsec:waveform_rank_angle_range_doppler_separability}
%To characterize whether different candidate angles share the same normalized range-Doppler response, we relate the 3D ambiguity structure to the rank of the waveform matrix formed by all RE transmit vectors.

For subsequent analysis, the 2D RE index $(n,m)$ is rearranged as $b=n+N_cm$, where $b=0,\ldots,N_{\mathrm{RE}}-1$ and $N_{\mathrm{RE}}=N_cN_s$. Accordingly, $n_b=b\bmod N_c$, $m_b=\lfloor b/N_c\rfloor$, and $\mathbf x_b\triangleq\mathbf x_{n_b,m_b}$. Collecting all RE transmit vectors gives $\mathbf X_{\mathrm{RE}}=[\mathbf x_0,\mathbf x_1,\ldots,\mathbf x_{N_{\mathrm{RE}}-1}]\in\mathbb C^{N_t\times N_{\mathrm{RE}}}$.
Its column space $\mathcal S_T=\operatorname{col}(\mathbf X_{\mathrm{RE}})$ is the transmit subspace occupied by the MIMO-OFDM waveform. Besides, the waveform rank $r=\operatorname{rank}(\mathbf X_{\mathrm{RE}})=\dim(\mathcal S_T)$ equals the minimum number of independent spatial basis vectors required to represent all RE transmit vectors.

Consider a candidate angle grid $\Theta=\{\theta_g\}_{g=1}^{N_\theta}$, and define $\mathbf a_0\triangleq\mathbf a_T(\theta_0;\mathbf d_T)$, $\mathbf a_g\triangleq\mathbf a_T(\theta_g;\mathbf d_T)$, while
$\mathbf A_\Theta=[\mathbf a_1,\ldots,\mathbf a_{N_\theta}]$.
Define the ambiguity matrix $\boldsymbol\Xi_{\theta_0}\in\mathbb C^{N_\theta\times N_{\rm RE}}$ with entries $[\boldsymbol\Xi_{\theta_0}]_{g,\iota(l,\nu)}=\chi_{\theta_0}(\theta_g,l,\nu)$, where $\iota(l,\nu)=l+N_c\nu$. Then, \eqref{eq:discrete_periodic_3d_af} can be written in matrix form as\vspace{-0.1cm}
\begin{equation}
\label{eq:ambiguity_sample_matrix}
\boldsymbol\Xi_{\theta_0}
=
\mathbf D_\rho
\mathbf A_\Theta^T
\mathbf X_{\mathrm{RE}}^{*}
\mathbf D_\alpha
\mathbf F_{\mathrm{RD}},
\end{equation}
where $\mathbf D_\rho = \operatorname{diag} \left(\rho_R(\theta_1,\theta_0),\ldots,\rho_R(\theta_{N_\theta},\theta_0)\right)$, $\mathbf D_\alpha = \operatorname{diag} \left(\mathbf a_0^H\mathbf x_0,\ldots,\mathbf a_0^H\mathbf x_{N_{\rm RE}-1}\right)$, and $\mathbf F_{\rm RD}\in\mathbb C^{N_{\rm RE}\times N_{\rm RE}}$ is the complete range-Doppler Fourier transform matrix with entries\vspace{-0.1cm}
\begin{equation}
\label{eq:range_doppler_transform_matrix}
[\mathbf F_{\mathrm{RD}}]_{b,\iota(l,\nu)}
=
e^{-j2\pi ln_b/N_c}
e^{j2\pi\nu m_b/N_s},
\end{equation}
The columns of $\mathbf F_{\rm RD}$ form a complete 2D periodic Fourier basis, so $\mathbf F_{\rm RD}$ is invertible. The $g$-th row of $\boldsymbol\Xi_{\theta_0}$ corresponds to the range-Doppler slice at $\theta_g$. Therefore, if a nonzero 3D ambiguity function is separable between the angular dimension and the range-Doppler plane, all rows differ only by complex scaling factors and $\operatorname{rank}(\boldsymbol\Xi_{\theta_0})=1$. The factorization in \eqref{eq:ambiguity_sample_matrix} further gives\vspace{-0.1cm}
\begin{equation}
\label{eq:ambiguity_rank_upper_bound}
\operatorname{rank}(\boldsymbol\Xi_{\theta_0})\le
\operatorname{rank}(\mathbf X_{\mathrm{RE}})
=r.
\end{equation}
Therefore, the rank of the ambiguity matrix $\boldsymbol\Xi_{\theta_0}$ is upper bounded by the transmit waveform rank, while the invertibility of $\mathbf F_{\rm RD}$ ensures that the range-Doppler transform introduces no additional rank loss.

The ambiguity matrix $\boldsymbol\Xi_{\theta_0}$ retains the waveform rank under the following conditions:
\begin{itemize}
    \item \textit{Condition 1}: $\operatorname{rank}(\mathbf A_\Theta)=N_t$.

    \item \textit{Condition 2}: $\mathbf a_0^H\mathbf x_b \neq0, \quad b=0,\ldots,N_{\mathrm{RE}}-1$.

    \item \textit{Condition 3}: $\rho_R(\theta_g,\theta_0)\neq0, \quad g=1,\ldots,N_\theta$.
\end{itemize}
where \textit{Condition 1} ensures that the candidate angle grid spans the transmit space, \textit{Condition 2} excludes zero projection of any RE onto the target direction, and \textit{Condition 3} excludes receive array correlation nulls. Under the above conditions, $\mathbf D_\alpha$ and $\mathbf D_\rho$ are nonsingular. Together with the invertibility of $\mathbf F_{\rm RD}$, \eqref{eq:ambiguity_rank_upper_bound} yields\vspace{-0.1cm}
\begin{equation}
\label{eq:ambiguity_rank_equivalence}
\operatorname{rank}(\boldsymbol\Xi_{\theta_0})
=
\operatorname{rank}(\mathbf X_{\mathrm{RE}})
=r.
\end{equation}
Based on the above rank relationships, we further examine how the transmit waveform rank determines the separability between the angular dimension and the range-Doppler plane.

\begin{theorem}
\label{thm:rank_one_separability}
If $\operatorname{rank}(\mathbf X_{\mathrm{RE}})=1$, there exist a nonzero basis vector $\mathbf w\in\mathbb C^{N_t}$ and complex coefficients $c_b$ such that $\mathbf x_b=c_b\mathbf w$ for $b\in [0,N_{\mathrm{RE}}-1]$. The 3D ambiguity function can then be expressed as the product of angular and range-Doppler components, i.e.,\vspace{-0.1cm}
\begin{equation}
\label{eq:rank_one_af_factorization}
\chi_{\theta_0}(\theta,l,\nu)
=
G_\theta(\theta_0,\theta)
D_{\mathrm{RD}}(l,\nu),
\end{equation}
\begin{comment}
\begin{subequations}
\label{eq:rank_one_af_factorization}
\begin{align}
\chi_{\theta_0}(\theta,l,\nu)
&=
G_\theta(\theta_0,\theta)
D_{\mathrm{RD}}(l,\nu),
\label{eq:rank_one_af_product}
\\
G_\theta(\theta_0,\theta)
&=
\rho_R(\theta,\theta_0)
[\mathbf a_0^H\mathbf w]
[\mathbf w^H\mathbf a_T(\theta;\mathbf d_T)],
\label{eq:rank_one_angular_response}
\\
D_{\mathrm{RD}}(l,\nu)
&=
\sum_{b=0}^{N_{\mathrm{RE}}-1}
|c_b|^2
e^{-j2\pi ln_b/N_c}
e^{j2\pi\nu m_b/N_s}.
\label{eq:rank_one_range_doppler_response}
\end{align}
\end{subequations}
\end{comment}
where we have\vspace{-0.1cm}
\begin{equation}
G_\theta(\theta_0,\theta)=\rho_R(\theta,\theta_0)[\mathbf a_0^H\mathbf w][\mathbf w^H\mathbf a_T(\theta;\mathbf d_T)],
\end{equation}
and\vspace{-0.1cm} 
\begin{equation}
D_{\mathrm{RD}}(l,\nu)=\sum_{b=0}^{N_{\mathrm{RE}}-1}|c_b|^2e^{-j2\pi ln_b/N_c}e^{j2\pi\nu m_b/N_s}.
\end{equation} 
Thus, a rank-one waveform makes the angular dimension separable from the range-Doppler plane. Conversely, under Conditions 1-3, separability of a nonzero 3D ambiguity function implies $\operatorname{rank}(\mathbf X_{\mathrm{RE}})=1$.
\end{theorem}
\begin{proof}
See \textbf{App.} B in the supplemental document \cite{mao2026detailed}.
%Please refer to Appendix~\hyperref[app:appendix_b]{B}.
\end{proof}

\begin{remark}
\label{remark:rank_one_separability}
Under a rank-one waveform, the candidate angle changes only the complex scale of the 3D ambiguity function, while the normalized range-Doppler shape remains unchanged. If Conditions 1-3 do not hold, some transmit dimensions may not be captured by the ambiguity matrix $\boldsymbol\Xi_{\theta_0}$, and a higher-rank waveform may therefore appear rank deficient or even separable.
\end{remark}

\subsection{Impact Analysis of MA Positions}
\label{subsec:ma_role_under_different_waveform_ranks}
Building on Theorem~\ref{thm:rank_one_separability}, this subsection further investigates how the MA positions affect the 3D ambiguity function under different waveform ranks.

\subsubsection{Rank-One Waveforms}
For $r=\operatorname{rank}(\mathbf X_{\mathrm{RE}})=1$, Theorem~\ref{thm:rank_one_separability} shows that all RE transmit vectors share the same 1D transmit subspace, i.e., $\mathbf x_b=c_b\mathbf w$. Since this decomposition has a scaling ambiguity, we set $\|\mathbf w\|_2=1$ without loss of generality. The total transmit energy across all $N_t$ transmit antennas on the $b$-th RE then satisfies\vspace{-0.1cm}
\begin{equation}
\label{eq:rank_one_re_energy}
\|\mathbf x_b\|_2^2
=
|c_b|^2\|\mathbf w\|_2^2
=
|c_b|^2.
\end{equation}
Therefore, $\mathbf w$ describes the relative amplitude and phase structure across the transmit antennas shared by all REs, whereas $|c_b|^2$ denotes the multi-antenna transmit energy on the $b$-th RE. In other words, a rank-one MIMO-OFDM waveform maps a set of scalar time-frequency coefficients $\{c_b\}$ onto the transmit array through a common spatial basis vector. Furthermore, according to the separable representation in Theorem~\ref{thm:rank_one_separability}, for any candidate angle satisfying $\chi_{\theta_0}(\theta,0,0)\neq0$, we have\vspace{-0.1cm}
\begin{equation}
\label{eq:normalized_rank_one_range_doppler_response}
\frac{\chi_{\theta_0}(\theta,l,\nu)}
{\chi_{\theta_0}(\theta,0,0)}
=
\frac{D_{\mathrm{RD}}(l,\nu)}
{D_{\mathrm{RD}}(0,0)}.
\end{equation}
Hence, the normalized range-Doppler ambiguity function is independent of the candidate angle. For fixed
$\mathbf w$ and $\{c_b\}$, the MA positions appear only in
$G_\theta(\theta_0,\theta)$, whereas $D_{\mathrm{RD}}(l,\nu)$ is independent of $\mathbf d_T$. Therefore, changing the MA positions cannot alter the normalized range-Doppler shape of a rank-one waveform.

Moreover, $D_{\mathrm{RD}}(l,\nu)$ depends only on the RE energy sequence $\{c_b\}$, not on the phases of $c_b$.
When all REs have equal energy, complete period Fourier orthogonality can make every nonzero range-Doppler cell vanish \cite{zhang2024cross_domain}. This property leads to the following rank-one waveform construction.

%Thus, under a rank-one waveform, different candidate angles exhibit the same normalized two-dimensional range-Doppler response, and the angle changes only its complex scale through $G_\theta(\theta_0,\theta)$. For given $\mathbf w$ and $\{c_b\}$, the MA positions appear only in $G_\theta(\theta_0,\theta)$, whereas $D_{\mathrm{RD}}(l,\nu)$ is independent of $\mathbf d_T$. Therefore, for a fixed rank-one waveform, adjusting the MA positions can change the angular response but cannot alter the normalized shape of the two-dimensional range-Doppler response.

%Moreover, \eqref{eq:rank_one_range_doppler_response} shows that the range-Doppler response of a rank-one waveform is determined solely by the RE energy sequence $\{|c_b|^2\}$ and is independent of the phases of $c_b$. In particular, when all REs have the same energy, complete-period Fourier orthogonality makes the response at every nonzero range-Doppler cell exactly zero. This property indicates that, for a given reference target, preserving the angular response at zero range and zero Doppler does not necessarily require a high-rank transmit subspace. The following theorem provides a more general rank-one compression result.

\begin{theorem}
\label{thm:rank_one_compression}
For fixed $\mathbf d_T$, consider any transmit waveform $\mathbf X_{\mathrm{RE}}$ and let $\mathbf R_\Sigma=\mathbf X_{\mathrm{RE}}\mathbf X_{\mathrm{RE}}^H$. If $\mathbf a_0^H\mathbf R_\Sigma\mathbf a_0>0$, a rank-one transmit waveform $\mathbf X_{\mathrm{RE}}^{(1)}=[\mathbf x_0^{(1)},\ldots,\mathbf x_{N_{\mathrm{RE}}-1}^{(1)}]$ can always be constructed as\vspace{-0.1cm}
\begin{equation}
\label{eq:rank_one_waveform_construction}
\mathbf x_b^{(1)}
=
\frac{
e^{j\phi_b}
\mathbf R_\Sigma\mathbf a_0
}{
\sqrt{N_{\mathrm{RE}}\,
\mathbf a_0^H\mathbf R_\Sigma\mathbf a_0}
},
\quad
b=0,\ldots,N_{\mathrm{RE}}-1,
\end{equation}
where $\phi_b\in[0,2\pi)$ is arbitrary. Without increasing the total transmit energy, this construction preserves the complete angular response of the original waveform at $(l,\nu)=(0,0)$ and makes the 3D ambiguity function zero at every nonzero range-Doppler cell of the complete periodic grid.
\end{theorem}
\begin{proof}
See \textbf{App.} C in the supplemental document \cite{mao2026detailed}.
%Please refer to Appendix~\hyperref[app:appendix_c]{C}.
\end{proof}

\begin{remark}
\label{remark:rank_one_compressibility}
Under the single reference target and complete periodic model considered herein, the zero range and zero Doppler angular response can always be reproduced by a one dimensional transmit subspace without increasing the transmit energy. Hence, a higher waveform rank is not essential for preserving this response, although the optimal rank structure may differ for multiple targets or other range-Doppler models.
\end{remark}

\subsubsection{Higher-Rank Waveforms}
For $r=\operatorname{rank}(\mathbf X_{\mathrm{RE}})\ge2$, let $\{\mathbf w_p\}_{p=1}^{r}$ be an orthonormal basis of $\mathcal S_T$, so that $\mathbf x_b=\sum_{p=1}^{r}c_{p,b}\mathbf w_p$. Substituting this expansion into \eqref{eq:discrete_periodic_3d_af} gives\vspace{-0.2cm}
\begin{equation}
\label{eq:high_rank_af_decomposition}
\begin{aligned}
\chi_{\theta_0}(\theta,l,\nu)
=&
\rho_R(\theta,\theta_0)
\sum_{p=1}^{r}
\sum_{q=1}^{r}
[\mathbf a_0^H\mathbf w_p]
\\[-1mm]
&\times
[\mathbf w_q^H\mathbf a_T(\theta;\mathbf d_T)]
D_{p,q}(l,\nu),
\end{aligned}
\end{equation}
where
$D_{p,q}(l,\nu)
=
\sum_{b=0}^{N_{\mathrm{RE}}-1}
c_{p,b}c_{q,b}^*
e^{-j2\pi ln_b/N_c}
e^{j2\pi\nu m_b/N_s}$.
It is noted that \eqref{eq:high_rank_af_decomposition} contains up to $r^2$ range-Doppler components. The diagonal terms $D_{p,p}(l,\nu)$ are determined by the RE energy sequences of individual spatial components, whereas $D_{p,q}(l,\nu)$, $p\neq q$, describe their cross correlations. The MA positions can change the angle-dependent weights $(\mathbf a_0^H\mathbf w_p)
[\mathbf w_q^H\mathbf a_T(\theta;\mathbf d_T)]$.
When the corresponding range-Doppler components have different structures, their superposition varies with both the candidate angle and the MA positions. Therefore, unlike the rank-one case, a higher-rank waveform allows the MA positions to further shape the range-Doppler ambiguity function across candidate angles.

\section{Waveform Design and Antenna Position Optimization}\label{section:joint_design}
% Compact subsection spacing only within the optimization section.
\makeatletter
\let\optimizationSavedSubsection\subsection
\let\optimizationSavedSubsubsection\subsubsection
\renewcommand{\subsection}{\@startsection{subsection}{2}{\z@}%
  {1.5ex plus 0.2ex minus 0.2ex}{0.7ex plus 0.1ex}%
  {\normalfont\normalsize\itshape}}
\renewcommand{\subsubsection}{\@startsection{subsubsection}{3}{\parindent}%
  {0ex}{0ex}{\normalfont\normalsize\itshape}}
\makeatother
% Keep positive, bounded gaps around the dense optimization equations.
\edef\optimizationSavedDisplaySkips{%
  \abovedisplayskip=\the\abovedisplayskip\relax
  \belowdisplayskip=\the\belowdisplayskip\relax
  \abovedisplayshortskip=\the\abovedisplayshortskip\relax
  \belowdisplayshortskip=\the\belowdisplayshortskip\relax}
\setlength{\abovedisplayskip}{4.5pt plus 0.5pt minus 0.5pt}
\setlength{\belowdisplayskip}{4.5pt plus 0.5pt minus 0.5pt}
\setlength{\abovedisplayshortskip}{0pt plus 0.5pt}
\setlength{\belowdisplayshortskip}{4.5pt plus 0.5pt minus 0.5pt}
Building on the 3D ambiguity analysis in Section~\ref{section:3D_AF_analysis}, this section jointly designs the transmit MA positions and SLP waveform for 3D ambiguity sidelobe suppression. We first formulate a 3D ISLR minimization problem and develop a PDD-AO algorithm for the general ISAC case, which accommodates higher-rank waveforms, and then derive an equivalent low-dimensional formulation for the radar-only case based on the rank-one optimality result.

\subsection{Problem Formulation}
\label{subsec:problem_formulation}
The discrete periodic 3D ambiguity function in \eqref{eq:discrete_periodic_3d_af} can be written in quadratic form as\vspace{-0.1cm}
\begin{equation}
\chi_{\theta_0}(\theta_g,l,\nu)
=
\mathbf x^H
\rho_g\widetilde{\mathbf A}_g(\mathbf d_T)\left(\mathbf D_\nu\otimes\mathbf D_l^*\right)\widetilde{\mathbf A}_0^H(\mathbf d_T)
\mathbf x,
\label{eq:quadratic_3d_af}
\end{equation}
where $\mathbf x=\operatorname{vec}(\mathbf X_{\mathrm{RE}})\in\mathbb C^{N_tN_{\rm RE}}$, $\rho_g=\rho_R(\theta_g,\theta_0)$, $\widetilde{\mathbf A}_0(\mathbf d_T)\triangleq \mathbf I_{N_{\mathrm{RE}}}\otimes\mathbf a_0\in\mathbb C^{N_tN_{\mathrm{RE}}\times N_{\mathrm{RE}}}$, and $\widetilde{\mathbf A}_g(\mathbf d_T)\triangleq \mathbf I_{N_{\mathrm{RE}}}\otimes\mathbf a_g\in\mathbb C^{N_tN_{\mathrm{RE}}\times N_{\mathrm{RE}}}$. Moreover, $\mathbf D_l \triangleq \operatorname{diag}\!\left(1,e^{j2\pi l/N_c},\ldots, e^{j2\pi l(N_c-1)/N_c}\right)\in\mathbb C^{N_c\times N_c}$ and $\mathbf D_\nu \triangleq \operatorname{diag}\!\left(1,e^{j2\pi\nu/N_s},\ldots, e^{j2\pi\nu(N_s-1)/N_s}\right)\in\mathbb C^{N_s\times N_s}$ are the range and Doppler phase matrices, respectively.

As shown in Section~\ref{section:3D_AF_analysis}, a multi-dimensional transmit subspace allows the range-Doppler ambiguity function to vary with the candidate angle. Therefore, evaluating only the 2D range-Doppler sidelobes at the target direction cannot fully characterize the 3D ambiguity leakage. Based on this observation, we extend the conventional 2D ISLR metric to the complete angle-range-Doppler grid and define the 3D ISLR as\vspace{-0.1cm}
\begin{equation}
\begin{aligned}
&\operatorname{ISLR}_{3\mathrm D}(\mathbf x,\mathbf d_T) \\[-1mm]
&\triangleq
\frac{\operatorname{ISL}_{3\mathrm D}(\mathbf x,\mathbf d_T)}
{\operatorname{ML}_{3\mathrm D}(\mathbf x,\mathbf d_T)}
=
\frac{\sum_{(\theta_g,l,\nu)\in\Omega_{\mathrm{SL}}^{(3\mathrm D)}}
\left|\chi_{\theta_0}(\theta_g,l,\nu)\right|^2}
{\sum_{(\theta_g,l,\nu)\in\Omega_{\mathrm{ML}}^{(3\mathrm D)}}
\left|\chi_{\theta_0}(\theta_g,l,\nu)\right|^2},
\label{eq:3d_islr}
\end{aligned}
\end{equation}
where $\mathcal M_\theta\subseteq\Theta$ denotes the angular mainlobe set containing $\theta_0$, $\mathcal S_\theta=\Theta\setminus\mathcal M_\theta$, while $\mathcal G_{\mathrm{RD}}=\{(l,\nu)\mid 0\le l<N_c,\ 0\le\nu<N_s\}$. The corresponding 3D mainlobe and sidelobe regions are $\Omega_{\mathrm{ML}}^{(3\mathrm D)}=\mathcal M_\theta\times\{(0,0)\}$ and $\Omega_{\mathrm{SL}}^{(3\mathrm D)}=\left[\mathcal S_\theta\times\{(0,0)\}\right]\cup\left[\Theta\times\left(\mathcal G_{\mathrm{RD}}\setminus\{(0,0)\}\right)\right]$.

Furthermore, since the 3D ISLR is invariant to a common scaling of the transmit waveform, the target illumination is constrained separately. By Parseval's identity, the cumulative illumination in the target direction can be written as\vspace{-0.1cm}
\begin{equation}
\begin{aligned}
P_{\mathrm I}(\mathbf x,\mathbf d_T)
\triangleq&
\sum_{m=0}^{N_s-1}\sum_{q=0}^{N_c-1}
\left|\mathbf a_0^H\widetilde{\mathbf x}_{q,m}\right|^2
=
\sum_{b=0}^{N_{\mathrm{RE}}-1}
\left|\mathbf a_0^H\mathbf x_b\right|^2
\\[-1mm]
=&
\mathbf x^H\widetilde{\mathbf A}_0(\mathbf d_T)
\widetilde{\mathbf A}_0^H(\mathbf d_T)\mathbf x.
\end{aligned}
\label{eq:target_illumination_power}
\end{equation}

Thus, the joint design problem is formulated as\vspace{-0.1cm}
\begin{subequations}\label{prob:joint_design}
    \begin{align}        
        \min_{\mathbf x,\mathbf d_T} \quad &  \operatorname{ISLR}_{3\mathrm D}(\mathbf x,\mathbf d_T) \label{P0 original objective function} \\[-0.2cm]
        \label{cons: CI}
        \mathrm{s.t.} \quad 
        & \Re\!\left\{\widetilde{\mathbf h}_{b,k'}^H(\mathbf d_T)\mathbf x_b\right\}\ge\gamma_{k'}, \nonumber \\[-1mm]
        &\quad b=0,\ldots,N_{\mathrm{RE}}-1, \quad k'=1,\ldots,2K, \\
        \label{cons: illumination}
        & P_{\mathrm I}(\mathbf x,\mathbf d_T)\ge P_{\mathrm{ill}}, \\
        \label{cons: power}
        & \|\mathbf x\|_2^2\le P_{\max}, \\
        \label{cons: MA position}
        & \mathbf d_T\in\mathcal D_T,
    \end{align}
\end{subequations}
where $P_{\rm ill}>0$ denotes the minimum target illumination threshold. Note that problem \eqref{prob:joint_design} is highly nonconvex. For fixed $\mathbf d_T$, the 3D ISLR is a ratio of quartic functions of $\mathbf x$, while \eqref{cons: illumination} is a nonconvex quadratic lower bound constraint. For fixed $\mathbf x$, the MA positions enter both the objective and the CI constraints through nonlinear spatial phases. Directly maintaining \eqref{cons: CI} and \eqref{cons: illumination} during AO can severely restrict the feasible MA position update. We therefore adopt the PDD framework \cite{shi2020pdd}, which converts these inequalities into equality constraints with nonnegative auxiliary variables, while \eqref{cons: power} and \eqref{cons: MA position} remain explicitly enforced.

\subsection{Problem Transformation}
\label{subsec:augmented_lagrangian_reformulation}
To incorporate the CI and target illumination constraints in problem \eqref{prob:joint_design} into the PDD framework, define the corresponding constraint margins as\vspace{-0.1cm}
\begin{equation}
\psi_{b,k'}(\mathbf x_b,\mathbf d_T)
\triangleq
\Re\!\left\{\widetilde{\mathbf h}_{b,k'}^H(\mathbf d_T)\mathbf x_b\right\}-\gamma_{k'},
\label{eq:ci_constraint_margin}
\end{equation}
and\vspace{-0.1cm}
\begin{equation}
\psi_{\mathrm I}(\mathbf x,\mathbf d_T)
\triangleq
P_{\mathrm I}(\mathbf x,\mathbf d_T)-P_{\mathrm{ill}}.
\label{eq:illumination_margin}
\end{equation}
The original inequalities are therefore equivalent to $\psi_{b,k'}\ge0$ and $\psi_{\mathrm I}\ge0$, respectively. Introducing the auxiliary variables $u_{b,k'}$ and $u_{\mathrm I}$, these inequalities can be equivalently expressed as $u_{b,k'}=\psi_{b,k'}$ and $u_{\mathrm I}=\psi_{\mathrm I}$, with $u_{b,k'}\ge0$ and $u_{\mathrm I}\ge0$. Hence, problem \eqref{prob:joint_design} can be equivalently reformulated as\vspace{-0.1cm}
\begin{subequations}\label{prob:pdd_equality_reformulation}
    \begin{align}        
        \min_{\mathbf x,\mathbf d_T,\{u_{b,k'}\},u_{\mathrm I}} \quad &  \operatorname{ISLR}_{3\mathrm D}(\mathbf x,\mathbf d_T) \label{P1 original objective function} \\[-0.2cm]
        \label{cons: CI margin}
        \mathrm{s.t.} \quad 
        & u_{b,k'}=\psi_{b,k'}(\mathbf x_b,\mathbf d_T), \nonumber \\[-1mm]
        &\quad b=0,\ldots,N_{\mathrm{RE}}-1, \quad k'=1,\ldots,2K, \\
        \label{cons: illumination margin}
        & u_{\mathrm I}=\psi_{\mathrm I}(\mathbf x,\mathbf d_T), \\
        & \eqref{cons: power}, \eqref{cons: MA position},
    \end{align}
\end{subequations}
For a given penalty parameter $\varrho>0$ and dual vector $\boldsymbol{\lambda}\in\mathbb R^{2KN_{\mathrm{RE}}+1}$, the equality constraints in problem \eqref{prob:pdd_equality_reformulation} are incorporated into the following augmented Lagrangian (AL) inner problem:\vspace{-0.1cm}
\begin{subequations}\label{prob:AL_inner_problem}
    \begin{align}        
        \min_{\mathbf x,\mathbf d_T,\mathbf u} \quad &  \operatorname{ISLR}_{3\mathrm D}(\mathbf x,\mathbf d_T)+\mathcal L_{\varrho}(\mathbf u,\boldsymbol{\lambda}) \label{P2 original objective function} \\[-0.2cm]
        \mathrm{s.t.} \quad 
        & \mathbf u\succeq\mathbf 0, \eqref{cons: power}, \eqref{cons: MA position},
    \end{align}
\end{subequations}
where $\mathbf u\in\mathbb R_+^{2KN_{\mathrm{RE}}+1}$ and $\boldsymbol{\psi}(\mathbf x,\mathbf d_T)\in\mathbb R^{2KN_{\mathrm{RE}}+1}$ collect all auxiliary variables and the corresponding constraint margins, respectively, while $\mathcal L_{\varrho}(\mathbf u,\boldsymbol{\lambda})=\boldsymbol{\lambda}^T\left[\mathbf u-\boldsymbol{\psi}(\mathbf x,\mathbf d_T)\right]+\frac{1}{2\varrho}\left\|\mathbf u-\boldsymbol{\psi}(\mathbf x,\mathbf d_T)\right\|_2^2$. We adopt the reciprocal penalty parameterization \(1/(2\varrho)\), so a smaller \(\varrho\) imposes a stronger penalty on the equality residual. For fixed \((\mathbf x,\mathbf d_T)\), the auxiliary variables are decoupled and admit a closed form solution.

\begin{proposition}
\label{prop:auxiliary_variable_solution}
For any fixed $(\mathbf x,\mathbf d_T)$, $\varrho>0$, and $\boldsymbol{\lambda}$, the unique optimal solution of \eqref{prob:AL_inner_problem} with respect to $\mathbf u\succeq\mathbf 0$ is\vspace{-0.1cm}
\begin{equation}
\mathbf u^\star
=
\left[
\boldsymbol{\psi}(\mathbf x,\mathbf d_T)
-\varrho\boldsymbol{\lambda}
\right]_+,
\label{eq:optimal_auxiliary_variables}
\end{equation}
where $[\,\cdot\,]_+$ denotes elementwise projection onto the nonnegative real axis.
\end{proposition}
\begin{proof}
See \textbf{App.} D in the supplemental document \cite{mao2026detailed}.
%Please refer to Appendix~\hyperref[app:appendix_d]{D}.
\end{proof}

Substituting \eqref{eq:optimal_auxiliary_variables} into \eqref{P2 original objective function} and omitting the term independent of \((\mathbf x,\mathbf d_T)\), problem \eqref{prob:AL_inner_problem} reduces to\vspace{-0.1cm}
\begin{subequations}\label{prob:reduced_inner_problem}
    \begin{align}        
        \min_{\mathbf x,\mathbf d_T} \quad &  \Phi_{\varrho,\boldsymbol{\lambda}}(\mathbf x,\mathbf d_T) \label{P3 original objective function} \\[-0.2cm]
        \mathrm{s.t.} \quad 
        & \eqref{cons: power}, \eqref{cons: MA position},
    \end{align}
\end{subequations}
where\vspace{-0.2cm}
\begin{equation}
\Phi_{\varrho,\boldsymbol{\lambda}}(\mathbf x,\mathbf d_T)
=
\operatorname{ISLR}_{3\mathrm D}(\mathbf x,\mathbf d_T)
+\frac{\varrho}{2}
\left\|
\left[
\boldsymbol{\lambda}
-\frac{\boldsymbol{\psi}(\mathbf x,\mathbf d_T)}{\varrho}
\right]_+
\right\|_2^2.
\label{eq:reduced_inner_objective}
\end{equation}
Problem \eqref{prob:reduced_inner_problem} contains only the two variable blocks $\mathbf x$ and $\mathbf d_T$, while the eliminated auxiliary variables can be recovered from \eqref{eq:optimal_auxiliary_variables} for evaluating the PDD residual.

\subsection{AO Inner Loop Update}
\label{subsec:inner_alternating_optimization}
For fixed PDD parameters $(\varrho,\boldsymbol{\lambda})$, we solve problem \eqref{prob:reduced_inner_problem} by alternately updating $\mathbf x$ and $\mathbf d_T$ using projected gradient descent (PGD) with backtracking line search.

\subsubsection{Transmit Waveform Update}
\label{subsubsec:transmit_waveform_update}
For fixed $\mathbf d_T^{(r)}$, the waveform subproblem can be formulated as\vspace{-0.1cm}
\begin{subequations}\label{prob:waveform_update_subproblem}
    \begin{align}        
        \min_{\mathbf x} \quad &  \Phi_{\varrho,\boldsymbol{\lambda}}
\left(\mathbf x,\mathbf d_T^{(r)}\right) \label{P4 original objective function} \\[-0.2cm]
        \mathrm{s.t.} \quad 
        & \eqref{cons: power}.
    \end{align}
\end{subequations}
Since $\mathbf x$ is complex-valued whereas $\Phi_{\varrho,\boldsymbol{\lambda}}$ is real-valued, we use Wirtinger calculus to derive the waveform gradient. For the convenience of derivation, let $\chi_{g,l,\nu}\triangleq\chi_{\theta_0}(\theta_g,l,\nu)$. From \eqref{eq:quadratic_3d_af}, we obtain\vspace{-0.1cm}
\begin{equation}
\nabla_{\mathbf x^*}|\chi_{g,l,\nu}|^2
=
\chi_{g,l,\nu}^*\mathbf Q_{g,l,\nu}\mathbf x
+\chi_{g,l,\nu}\mathbf Q_{g,l,\nu}^H\mathbf x,
\label{eq:af_sample_waveform_gradient}
\end{equation}
where $\mathbf Q_{g,l,\nu}=\rho_g\widetilde{\mathbf A}_g(\mathbf D_\nu\otimes\mathbf D_l^*)\widetilde{\mathbf A}_0^H $. Applying the quotient rule gives\vspace{-0.1cm}
\begin{equation}
\begin{aligned}
\nabla_{\mathbf x^*}\operatorname{ISLR}_{3\mathrm D}
=&
\frac{1}{\operatorname{ML}_{3\mathrm D}^{2}}
\Bigg[
\operatorname{ML}_{3\mathrm D}
\sum_{(\theta_g,l,\nu)\in\Omega_{\mathrm{SL}}^{(3\mathrm D)}}
\nabla_{\mathbf x^*}|\chi_{g,l,\nu}|^2
\\[-1mm]
&-
\operatorname{ISL}_{3\mathrm D}
\sum_{(\theta_g,l,\nu)\in\Omega_{\mathrm{ML}}^{(3\mathrm D)}}
\nabla_{\mathbf x^*}|\chi_{g,l,\nu}|^2
\Bigg].
\end{aligned}
\label{eq:islr_waveform_gradient}
\end{equation}
Using $\nabla_{\mathbf x_b^*}\psi_{b,k'}=\frac{1}{2}\widetilde{\mathbf h}_{b,k'}(\mathbf d_T)$ and $\nabla_{\mathbf x_b^*}\psi_{\mathrm I}=\mathbf a_0\mathbf a_0^H\mathbf x_b$,
the gradient with respect to the $b$-th RE waveform is given by\vspace{-0.1cm}
\begin{equation}
\begin{aligned}
&\nabla_{\mathbf x_b^*}
\Phi_{\varrho,\boldsymbol{\lambda}}
=
\left[
\nabla_{\mathbf x^*}\operatorname{ISLR}_{3\mathrm D}
\right]_b \\[-1mm]
&-\frac{1}{2}\sum_{k'=1}^{2K}
\left[
\lambda_{b,k'}-\frac{\psi_{b,k'}}{\varrho}
\right]_+
\widetilde{\mathbf h}_{b,k'}(\mathbf d_T)
-
\left[
\lambda_{\mathrm I}-\frac{\psi_{\mathrm I}}{\varrho}
\right]_+
\mathbf a_0\mathbf a_0^H\mathbf x_b,
\end{aligned}
\label{eq:waveform_block_gradient}
\end{equation}
where $[\nabla_{\mathbf x^*}\operatorname{ISLR}_{3\mathrm D}]_b$ denotes the gradient block associated with the $b$-th RE. Using the gradient in \eqref{eq:waveform_block_gradient}, the transmit waveform is updated via PGD as\vspace{-0.1cm}
\begin{equation}
\mathbf x^{(r+1)}=\Pi_{\mathcal X}\left(\mathbf x^{(r)}-\frac{2}{\kappa_x^{(r)}}\nabla_{\mathbf x^*}\Phi_{\varrho,\boldsymbol{\lambda}}\right),
\label{eq:waveform_pgd_update}
\end{equation}
where $\mathcal X\triangleq\{\mathbf x:\|\mathbf x\|_2^2\le P_{\max}\}$, $\Pi_{\mathcal X}(\cdot)$ denotes the Euclidean projection onto $\mathcal X$, and $\kappa_x^{(r)}>0$ is the curvature parameter. The parameter \(\kappa_x^{(r)}\) is selected by backtracking such that the projected point in \eqref{eq:waveform_pgd_update} satisfies\vspace{-0.1cm}
\begin{equation}
\begin{aligned}
&
\Phi_{\varrho,\boldsymbol{\lambda}}
\left(
\mathbf x^{(r+1)},\mathbf d_T^{(r)}
\right) \\[-1mm]
&\le
\Phi_{\varrho,\boldsymbol{\lambda}}
\left(
\mathbf x^{(r)},\mathbf d_T^{(r)}
\right)
+
2\Re
\left\{
\left[
\nabla_{\mathbf x^*}
\Phi_{\varrho,\boldsymbol{\lambda}}
\right]^H 
\left(
\mathbf x^{(r+1)}-\mathbf x^{(r)}
\right)
\right\} \\[-1mm]
&\quad+
\frac{\kappa_x^{(r)}}{2}
\left\|
\mathbf x^{(r+1)}-\mathbf x^{(r)}
\right\|_2^2,
\end{aligned}
\label{eq:waveform_backtracking_condition}
\end{equation}
where $\nabla_{\mathbf x^*}\Phi_{\varrho,\boldsymbol{\lambda}}$ is evaluated at $(\mathbf x^{(r)},\mathbf d_T^{(r)})$. We further require $\operatorname{ML}_{3\mathrm D}(\mathbf x^{(r+1)},\mathbf d_T^{(r)})\ge\epsilon_{\mathrm{ML}}$, where $\epsilon_{\mathrm{ML}}>0$ prevents the 3D ISLR denominator from approaching zero.

\subsubsection{MA Position Update}
\label{subsubsec:ma_position_update}
With the latest waveform $\mathbf x^{(r+1)}$ fixed, the MA position subproblem can be formulated as\vspace{-0.1cm}
\begin{subequations}\label{prob:position_update_subproblem}
    \begin{align}        
        \min_{\mathbf d_T} \quad &  \Phi_{\varrho,\boldsymbol{\lambda}}
\left(\mathbf x^{(r+1)},\mathbf d_T\right) \label{P5 original objective function} \\[-0.2cm]
        \mathrm{s.t.} \quad 
        & \eqref{cons: MA position}.
    \end{align}
\end{subequations}
For brevity, $\mathbf x$ below denotes the fixed $\mathbf x^{(r+1)}$.
Since $d_{T,0}=0$ is fixed, only $p=1,\ldots,N_t-1$ needs to be updated. Moreover, the receive array is fixed, so $\rho_g$ is independent of $\mathbf d_T$.
The derivative of the ambiguity function value $\chi_{g,l,\nu}$ with respect to the $p$-th free MA position is given by\vspace{-0.1cm}
\begin{equation}
\begin{aligned}
\frac{\partial\chi_{g,l,\nu}}
{\partial d_{T,p}}
=&
\rho_g\mathbf x^H
\Bigg[
\frac{
\partial\widetilde{\mathbf A}_g(\mathbf d_T)
}{
\partial d_{T,p}
}
\left(
\mathbf D_\nu\otimes\mathbf D_l^*
\right)
\widetilde{\mathbf A}_0^H(\mathbf d_T)
\\[-1mm]
&+
\widetilde{\mathbf A}_g(\mathbf d_T)
\left(
\mathbf D_\nu\otimes\mathbf D_l^*
\right)
\frac{
\partial\widetilde{\mathbf A}_0^H(\mathbf d_T)
}{
\partial d_{T,p}
}
\Bigg]
\mathbf x,
\end{aligned}
\label{eq:af_sample_position_derivative}
\end{equation}
where $\frac{\partial\widetilde{\mathbf A}_g}{\partial d_{T,p}}=\mathbf I_{N_{\rm RE}}\otimes\frac{\partial\mathbf a_g}{\partial d_{T,p}}$, $\frac{\partial\widetilde{\mathbf A}_0}{\partial d_{T,p}}=\mathbf I_{N_{\rm RE}}\otimes\frac{\partial\mathbf a_0}{\partial d_{T,p}}$. The corresponding steering vector derivatives are $\frac{\partial\mathbf a_g}{\partial d_{T,p}}=jk_c\sin\theta_g\,\mathbf E_p\mathbf a_g$, $\frac{\partial\mathbf a_0}{\partial d_{T,p}}=jk_c\sin\theta_0\,\mathbf E_p\mathbf a_0$, where $\mathbf E_p\in\mathbb R^{N_t\times N_t}$ is diagonal with a one at its $(p+1,p+1)$-th entry and zeros elsewhere.

Based on \eqref{eq:af_sample_position_derivative}, the derivative of the corresponding squared magnitude $|\chi_{g,l,\nu}|^2$ is\vspace{-0.1cm}
\begin{equation}
\frac{\partial|\chi_{g,l,\nu}|^2}{\partial d_{T,p}}
=
2\Re\!\left\{
\chi_{g,l,\nu}^*
\frac{\partial\chi_{g,l,\nu}}{\partial d_{T,p}}
\right\}.
\label{eq:af_power_position_derivative}
\end{equation}
The derivative of $\operatorname{ISLR}_{3\mathrm D}$ with respect to $d_{T,p}$ is then obtained using the same quotient rule as in \eqref{eq:islr_waveform_gradient}. To construct the complete MA position gradient of $\Phi_{\varrho,\boldsymbol{\lambda}}$, we further derive the position derivatives of the CI and target illumination margins. For the $k'$-th CI margin on the $b$-th RE, we have\vspace{-0.1cm}
\begin{equation}
\begin{aligned}
\frac{
\partial\psi_{b,k'}
}{
\partial d_{T,p}
}
=&
\Re\Bigg\{
-jk_c\xi_{b,k'}
\sum_{\ell=1}^{L_k}
\beta_{k,\ell}^*
e^{j2\pi n_b\Delta f\tau_{k,\ell}}
\sin\varphi_{k,\ell}
\\[-1mm]
&\times
\mathbf a_T^H
(\varphi_{k,\ell};\mathbf d_T)
\mathbf E_p
\mathbf x_b
\Bigg\},
\end{aligned}
\label{eq:ci_margin_position_derivative}
\end{equation}
with $\xi_{b,k'}$ given by\vspace{-0.1cm}
\begin{equation}
\xi_{b,k'}
=
\begin{cases}
s_{n_b,m_b,k}^*(\sin\phi_M+j\cos\phi_M),
& k'=2k-1,
\\
s_{n_b,m_b,k}^*(\sin\phi_M-j\cos\phi_M),
& k'=2k,
\end{cases}
\label{eq:ci_phase_factor}
\end{equation}
where $k=\lceil k'/2\rceil$. Similarly, the derivative of the target illumination margin is expressed as\vspace{-0.1cm}
\begin{equation}
\frac{\partial\psi_{\mathrm I}}{\partial d_{T,p}}
=
2\Re\left\{\sum_{b=0}^{N_{\rm RE}-1}\left(\mathbf a_0^H\mathbf x_b\right)^*\left(\frac{\partial\mathbf a_0^H}{\partial d_{T,p}}\mathbf x_b\right)\right\}.
\label{eq:illumination_margin_position_derivative}
\end{equation}

Combining the above derivatives gives\vspace{-0.1cm}
\begin{equation}
\begin{aligned}
\frac{\partial\Phi_{\varrho,\boldsymbol{\lambda}}}
{\partial d_{T,p}}
=&
\frac{\partial\operatorname{ISLR}_{3\mathrm D}}
{\partial d_{T,p}}
-
\sum_{b=0}^{N_{\mathrm{RE}}-1}\sum_{k'=1}^{2K}
\left[
\lambda_{b,k'}-\frac{\psi_{b,k'}}{\varrho}
\right]_+
\frac{\partial\psi_{b,k'}}{\partial d_{T,p}}
\\[-1mm]
&-
\left[
\lambda_{\mathrm I}-\frac{\psi_{\mathrm I}}{\varrho}
\right]_+
\frac{\partial\psi_{\mathrm I}}{\partial d_{T,p}}.
\end{aligned}
\label{eq:position_block_gradient}
\end{equation}

After collecting the derivatives of all free position variables into $\nabla_{\mathbf d_T}\Phi_{\varrho,\boldsymbol{\lambda}}$, the MA positions are updated via PGD as\vspace{-0.1cm}
\begin{equation}
\begin{aligned}
\mathbf d_T^{(r+1)}
=
\Pi_{\mathcal D_T}\!\left(
\mathbf d_T^{(r)}
-\frac{1}{\kappa_d^{(r)}}
\nabla_{\mathbf d_T}
\Phi_{\varrho,\boldsymbol{\lambda}}
\right),
\end{aligned}
\label{eq:position_pgd_update}
\end{equation}
where $\Pi_{\mathcal D_T}(\cdot)$ denotes the Euclidean projection onto $\mathcal D_T$.
%Since $\mathcal D_T$ contains only the array-boundary, reference-position, and minimum-spacing constraints, this projection is a bounded isotonic regression problem. 
With $y_p=d_{T,p}-pd_{\min}$, this projection reduces to a bounded isotonic regression problem and can be computed exactly in linear time using the pool adjacent violators algorithm.

As in the waveform update, $\kappa_d^{(r)}>0$ is selected by backtracking such that the projected point in \eqref{eq:position_pgd_update} satisfies\vspace{-0.1cm}
\begin{equation}
\begin{aligned}
&\Phi_{\varrho,\boldsymbol{\lambda}}
\left(\mathbf x^{(r+1)},\mathbf d_T^{(r+1)}\right) \\[-1mm]
&\le
\Phi_{\varrho,\boldsymbol{\lambda}}
\left(\mathbf x^{(r+1)},\mathbf d_T^{(r)}\right)
+
\left[
\nabla_{\mathbf d_T}
\Phi_{\varrho,\boldsymbol{\lambda}}
\right]^T
\left(\mathbf d_T^{(r+1)}-\mathbf d_T^{(r)}\right)
\\[-1mm]
&\quad+
\frac{\kappa_d^{(r)}}{2}
\left\|
\mathbf d_T^{(r+1)}-\mathbf d_T^{(r)}
\right\|_2^2,
\end{aligned}
\label{eq:position_backtracking_condition}
\end{equation}
while also requiring
$\operatorname{ML}_{3\mathrm D}(\mathbf x^{(r+1)},\mathbf d_T^{(r+1)})\ge\epsilon_{\mathrm{ML}}$.

\begin{remark}
\label{remark:pdd_position_update}
PDD is especially useful for the MA position update. If the position-dependent CI and target illumination constraints are enforced explicitly at each AO iteration, fixing the waveform can substantially restrict the admissible update region of $\mathbf d_T$. By incorporating these coupled constraints into the AL objective, the position update only needs to explicitly enforce $\mathbf d_T\in\mathcal D_T$, while the original CI and illumination constraints are progressively enforced through the PDD outer updates. This provides greater freedom for adjusting the MA positions.
\end{remark}

\subsection{PDD Outer Loop Update}
\label{subsec:pdd_outer_loop_update}
Let $(\widehat{\mathbf x}^{(t)},\widehat{\mathbf d}_T^{(t)})$ denote the solution returned by the inner AO procedure at the $t$-th PDD outer iteration. Using the closed form auxiliary variable solution in \eqref{eq:optimal_auxiliary_variables}, the equality residual is directly evaluated as\vspace{-0.1cm}
\begin{equation}
\!\!\!\!\varepsilon_{\mathrm{eq}}^{(t)}\triangleq\left\|\left[\boldsymbol{\psi}\left(\widehat{\mathbf x}^{(t)},\widehat{\mathbf d}_T^{(t)}\right)\!-\!\varrho^{(t)}\boldsymbol{\lambda}^{(t)}\right]_+\!\!-\!\boldsymbol{\psi}\left(\widehat{\mathbf x}^{(t)},\widehat{\mathbf d}_T^{(t)}\right)\right\|_\infty.
\label{eq:pdd_equality_residual}
\end{equation}

Furthermore, let $\eta^{(t)}>0$ denote the residual threshold used to determine the PDD outer update. If $\varepsilon_{\mathrm{eq}}^{(t)}\le\eta^{(t)}$, the multiplier is updated as $\boldsymbol{\lambda}^{(t+1)}=\left[\boldsymbol{\lambda}^{(t)}-\frac{\boldsymbol{\psi}(\widehat{\mathbf x}^{(t)},\widehat{\mathbf d}_T^{(t)})}{\varrho^{(t)}}\right]_+$, while $\varrho^{(t+1)}=\varrho^{(t)}$, and the residual threshold is tightened as $\eta^{(t+1)}=c_\eta\eta^{(t)}$, where $0<c_\eta<1$. Otherwise, we set $\boldsymbol{\lambda}^{(t+1)}=\boldsymbol{\lambda}^{(t)}$, $\eta^{(t+1)}=\eta^{(t)}$, and $\varrho^{(t+1)}=c_\varrho\varrho^{(t)}$, where $0<c_\varrho<1$. Since the coefficient of the quadratic penalty term in the AL is $1/(2\varrho)$, decreasing $\varrho^{(t)}$ strengthens the penalty on equality constraint violations. Accordingly, the PDD outer loop switches between dual updates and penalty strengthening based on the current residual, progressively enforcing the original CI and target illumination constraints.

\subsection{Computational Complexity Analysis}
\label{subsec:computational_complexity}
\begin{algorithm}[t]
\caption{Proposed PDD-AO Algorithm}
\label{alg:pdd_ao}
\begin{algorithmic}[1]
\STATE Initialize $\mathbf x^{(0)}\in\mathcal X$, $\mathbf d_T^{(0)}\in\mathcal D_T$ with $\operatorname{ML}_{3\mathrm D}\ge\epsilon_{\mathrm{ML}}$, $\boldsymbol{\lambda}^{(0)}=\mathbf 0$, $\varrho^{(0)}>0$, $\eta^{(0)}>0$, $0<c_\varrho,c_\eta<1$, and $t\gets0$.
\REPEAT
    \REPEAT
        \STATE Update $\mathbf x$ by
        \eqref{eq:waveform_pgd_update}--\eqref{eq:waveform_backtracking_condition}.
        \STATE Update $\mathbf d_T$ by
        \eqref{eq:position_pgd_update}--\eqref{eq:position_backtracking_condition}.
    \UNTIL{the inner stopping criterion is satisfied}
    \STATE
    $(\widehat{\mathbf x}^{(t)},\widehat{\mathbf d}_T^{(t)})
    \gets(\mathbf x,\mathbf d_T)$,
    and compute
    $\varepsilon_{\mathrm{eq}}^{(t)}$
    by \eqref{eq:pdd_equality_residual}.
    \IF{$\varepsilon_{\mathrm{eq}}^{(t)}\le\eta^{(t)}$}
        \STATE
        $\boldsymbol{\lambda}^{(t+1)}
        \gets
        \left[
        \boldsymbol{\lambda}^{(t)}
        -
        \dfrac{
        \boldsymbol{\psi}
        (\widehat{\mathbf x}^{(t)},\widehat{\mathbf d}_T^{(t)})
        }{
        \varrho^{(t)}
        }
        \right]_+$,
        $\varrho^{(t+1)}\gets\varrho^{(t)}$,
        $\eta^{(t+1)}\gets c_\eta\eta^{(t)}$.
    \ELSE
        \STATE
        $\boldsymbol{\lambda}^{(t+1)}
        \gets\boldsymbol{\lambda}^{(t)}$,
        $\varrho^{(t+1)}
        \gets c_\varrho\varrho^{(t)}$,
        $\eta^{(t+1)}\gets\eta^{(t)}$.
    \ENDIF
    \STATE $t\gets t+1$.
\UNTIL{the outer stopping criterion is satisfied}
\STATE \textbf{Output:}
$\mathbf x^\star,\mathbf d_T^\star$.
\end{algorithmic}
\end{algorithm}
Algorithm~\ref{alg:pdd_ao} summarizes the proposed PDD-AO algorithm. The 2D Fourier structure over the RE grid enables batch evaluation of the 3D ambiguity function and its gradients over the angle-range-Doppler grid using fast Fourier transforms (FFTs) and inverse transforms. Accordingly, one waveform update requires $\mathcal O\left(N_\theta N_{\mathrm{RE}}[N_t+\log N_{\mathrm{RE}}]\right)$ operations for the 3D ISLR and its gradient and $\mathcal O(KN_tN_{\mathrm{RE}})$ operations for the CI terms. The MA position update reuses the same Fourier structure. With $L_h\triangleq\max_k L_k$, evaluating the position derivatives of the communication channels additionally requires $\mathcal O(KN_cN_tL_h)$ operations. Therefore, the dominant complexity of one inner AO iteration is $\mathcal O(N_\theta N_{\mathrm{RE}}\left[N_t+\log N_{\mathrm{RE}}\right]+KN_tN_{\mathrm{RE}}+KN_cN_tL_h)$. Letting $I_{\mathrm{out}}$ and $I_{\mathrm{in}}$ denote the number of PDD outer iterations and the average number of inner AO iterations, respectively, and treating the average number of backtracking steps as a bounded constant, the overall complexity is $\mathcal O\{I_{\mathrm{out}}I_{\mathrm{in}}[N_\theta N_{\mathrm{RE}}\left(N_t+\log N_{\mathrm{RE}}\right)+KN_tN_{\mathrm{RE}}+KN_cN_tL_h]\}$.

\subsection{Radar-Only Case}
\label{subsec:radar_only_case}
We also consider the radar-only counterpart of problem \eqref{prob:joint_design} by removing the communication CI constraint, yielding\vspace{-0.1cm}
\begin{subequations}\label{prob:radar_only}
    \begin{align}        
        \min_{\mathbf X_{\mathrm{RE}},\mathbf d_T} \quad &  \operatorname{ISLR}_{3\mathrm D}(\mathbf X_{\mathrm{RE}},\mathbf d_T) \label{P6 original objective function} \\[-0.2cm]
        \label{cons: X_RE power}
        \mathrm{s.t.} \quad 
        & \|\mathbf X_{\mathrm{RE}}\|_F^2\le P_{\max}, \\
        & \eqref{cons: illumination}, \eqref{cons: MA position}.
    \end{align}
\end{subequations}

Recall from Theorem~\ref{thm:rank_one_compression} that, for fixed $\mathbf d_T$, any waveform with nonzero target illumination can be replaced by a rank-one waveform that preserves the ambiguity function values at $(l,\nu)=(0,0)$ for all candidate angles and makes them zero at all $(l,\nu)\neq(0,0)$, without increasing the transmit energy. Since the 3D mainlobe of problem \eqref{prob:radar_only} is confined to $(l,\nu)=(0,0)$, this rank-one construction preserves feasibility and cannot increase the 3D ISLR of a globally optimal solution. This yields the following corollary.

\begin{corollary}[Rank-One Globally Optimal Structure of the Radar-Only Problem]
\label{cor:radar_only_rank_one}
Suppose that problem \eqref{prob:radar_only} is feasible. Then, there exists a globally optimal solution $\mathbf X_{\mathrm{RE}}^\star$ satisfying $\operatorname{rank}(\mathbf X_{\mathrm{RE}}^\star)=1$. Moreover, a globally optimal solution can be chosen with equal transmit energy across all REs.
\end{corollary}
\begin{proof}
See \textbf{App.} E in the supplemental document \cite{mao2026detailed}.
%Please refer to Appendix~\hyperref[app:appendix_e]{E}.
\end{proof}

Because the 3D ISLR is invariant to an overall waveform scaling, the rank-one solution in Corollary~\ref{cor:radar_only_rank_one} can be scaled to use the full transmit power. Together with its equal RE energy structure and the invariance of both the ambiguity function and target illumination to the RE phases, it can be written without loss of optimality as\vspace{-0.1cm}
\begin{equation}
\mathbf x_b
=
\sqrt{\frac{P_{\max}}{N_{\mathrm{RE}}}}\,\mathbf w,
\quad
b=0,\ldots,N_{\mathrm{RE}}-1,
\label{eq:radar_only_rank_one_waveform}
\end{equation}
where $\|\mathbf w\|_2=1$. Substituting \eqref{eq:radar_only_rank_one_waveform} into \eqref{eq:discrete_periodic_3d_af} yields\vspace{-0.1cm}
\begin{equation}
\!\!\!\!\chi_{\theta_0}(\theta_g,l,\nu)
\!=\!
\begin{cases}
P_{\max}\rho_g
(\mathbf a_0^H\mathbf w)
(\mathbf w^H\mathbf a_g),
& \!(l,\nu)\!=\!(0,0),
\\
0,
& \!(l,\nu)\!\neq\!(0,0).
\end{cases}
\label{eq:radar_only_ambiguity_response}
\end{equation}
Thus, the equal energy rank-one waveform can make the 3D ambiguity function zero at every range-Doppler grid cell with $(l,\nu)\neq(0,0)$.

Substituting \eqref{eq:radar_only_rank_one_waveform} and \eqref{eq:radar_only_ambiguity_response} into the 3D ISLR and canceling the common nonzero factor $P_{\max}^2|\mathbf a_0^H\mathbf w|^2$ yields the following equivalent low-dimensional formulation:\vspace{-0.1cm}
\begin{subequations}\label{prob:reduced_radar_only}
    \begin{align}        
        \min_{\mathbf w,\mathbf d_T} \quad & \frac{\mathbf w^H\mathbf B_{\mathrm{SL}}(\mathbf d_T)\mathbf w}{\mathbf w^H\mathbf B_{\mathrm{ML}}(\mathbf d_T)\mathbf w} \label{P7 original objective function} \\
        \mathrm{s.t.} \quad 
        & P_{\max}|\mathbf a_0^H\mathbf w|^2\ge P_{\mathrm{ill}},  \|\mathbf w\|_2=1, \eqref{cons: MA position}.
    \end{align}
\end{subequations}
where $\mathbf B_{\mathrm{SL}}(\mathbf d_T)\triangleq\sum_{\theta_g\in\mathcal S_\theta}|\rho_g|^2\mathbf a_g\mathbf a_g^H$ and $\mathbf B_{\mathrm{ML}}(\mathbf d_T)\triangleq\sum_{\theta_g\in\mathcal M_\theta}|\rho_g|^2\mathbf a_g\mathbf a_g^H$. Problem \eqref{prob:reduced_radar_only} reduces the waveform variable from $N_tN_{\mathrm{RE}}$ complex dimensions to $\mathbf w$ without changing the globally optimal value of \eqref{prob:radar_only}. It can be solved by adapting the preceding PDD-AO framework, where all CI terms are removed, only the target illumination PDD term is retained, and $\mathbf w$ and $\mathbf d_T$ are alternately updated. Although the rank-one reformulation preserves the global optimum, the resulting nonconvex PDD-AO solver converges only to a corresponding first order stationary point.

\begin{remark}
\label{remark:radar_only_rank_one_limitation}
The rank-one optimal structure is specific to the radar-only formulation. In the complete ISAC problem, the CI constraints on individual REs depend jointly on the instantaneous communication symbols and position-dependent channels and generally cannot be satisfied within a common one dimensional transmit subspace. The communication constraints can therefore require a higher dimensional transmit subspace, allowing the range-Doppler ambiguity function to vary with the candidate angle, as characterized in Section~\ref{section:3D_AF_analysis}.
\end{remark}

\optimizationSavedDisplaySkips
\let\subsection\optimizationSavedSubsection
\let\subsubsection\optimizationSavedSubsubsection
\section{Simulation Results and Analysis}\label{section:simulation}
% Reduce the positive gap between single-column floats and the text.
\setlength{\textfloatsep}{12pt plus 2pt minus 2pt}
In this section, we evaluate the proposed MA-empowered MIMO-OFDM ISAC system through MATLAB simulations. Unless otherwise specified, the main simulation parameters are summarized in Table~\ref{tab:simulation_parameters}. The communication links follow the position-dependent frequency selective geometric channel in Section~\ref{subsec:comm_model}, with three effective paths per user, relative path powers $[0,-4,-8]$ dB, a reference gain of $-30$ dB at $1$ m, and a path loss exponent of $2.6$. By default, $K=2$ communication users are located $45$ m and $75$ m from the BS, with excess path delays $[0,25,70]$ ns and $[0,40,85]$ ns, respectively. The one reference target angle is located at $\theta_0=10^\circ$, while the candidate angles are sampled over $[-90^\circ,90^\circ]$ with a spacing of $1^\circ$. The angular mainlobe set is $\mathcal M_\theta=\Theta\cap[\theta_0-3^\circ,\theta_0+3^\circ]$, and the corresponding 3D mainlobe region is $\Omega_{\mathrm{ML}}^{(3\mathrm D)}=\mathcal M_\theta\times\{(0,0)\}$. The normalized transmit power budget and illumination threshold are defined as $bar P_{\mathrm T}\triangleq P_{\max}/N_s$ and $\bar P_{\mathrm{ill}}\triangleq P_{\mathrm{ill}}/N_s$, respectively, and are set to $10$ and $8$ on a linear scale.

\begin{table}[t]
\caption{Simulation Parameters}
\label{tab:simulation_parameters}
\centering
\footnotesize
\renewcommand{\arraystretch}{1.05}
\begin{tabularx}{\columnwidth}{|c|X|c|}
\hline
Symbol & Parameter & Value \\
\hline
$f_c$ & Carrier frequency & $24$~GHz \\
\hline
$B$ & Signal bandwidth & $120$~MHz \\
\hline
$N_t,N_r$ & Number of transmit/receive antennas & $6$ \\
\hline
$D_T$ & MA movement-region length & $20\lambda$ \\
\hline
$d_{\min}$ & Minimum transmit-MA spacing & $\lambda/2$ \\
\hline
$d_R$ & Receive-antenna spacing & $\lambda/2$ \\
\hline
$N_c$ & Number of subcarriers & $32$ \\
\hline
$N_s$ & Number of OFDM symbols & $16$ \\
\hline
$M$ & PSK modulation order & $4$ \\
\hline
$K$ & Number of communication users & $2$ \\
\hline
$\Gamma$ & Communication QoS threshold & $6$~dB \\
\hline
$\sigma_c^2,\sigma_r^2$ & Communication/radar noise power & $-70$~dBm \\
\hline
\end{tabularx}
\end{table}

To validate the performance gains provided by the proposed joint MA position optimization and SLP waveform design, we consider the following benchmark schemes.
\begin{itemize}
\item \textbf{FPA SLP-OFDM baselines:} 1) \textbf{Compact FPA SLP-OFDM} uses a compact ULA with $\lambda/2$ spacing, 2) \textbf{Random FPA SLP-OFDM} uses a fixed nonuniform array satisfying the minimum spacing constraint within the available region, and 3) \textbf{Sparse FPA SLP-OFDM} uniformly distributes the $N_t$ antennas over the entire region of length $D_T$. All three baselines use the same $N_t$ transmit antennas and optimize only the SLP waveform with fixed antenna positions.

\item \textbf{Comm-only:} The sensing objective and illumination constraint are removed, and the minimum normalized CI margin over all users and REs is maximized:\vspace{-0.1cm}
\begin{subequations}\label{prob:communication_only}
    \begin{align}        
        \!\max_{\mathbf x,\mathbf d_T,t} \quad & t \\[-0.2cm]
        \mathrm{s.t.} \quad 
        & \Re\!\left\{\widetilde{\mathbf h}_{b,k'}^H(\mathbf d_T)\mathbf x_b\right\}\ge t\gamma_{k'}, \nonumber \\[-1mm]
        &\!\!\quad b=0,\ldots,N_{\mathrm{RE}}-1, \quad k'=1,\ldots,2K, \\
        & t\ge0, \eqref{cons: power}, \eqref{cons: MA position},
    \end{align}
\end{subequations}
where $t$ is the minimum normalized CI margin. Problem \eqref{prob:communication_only} can still be solved using the proposed PDD-AO framework by removing the sensing related terms and jointly updating the transmit waveform, MA positions, and margin variable $t$.

\item \textbf{Radar-only:} The reduced formulation in Section~\ref{subsec:radar_only_case} is adopted by jointly optimizing $\mathbf w$ and $\mathbf d_T$ without communication constraints.

\item \textbf{MA ZF-OFDM:} The MA positions are optimized as in the proposed scheme, while the SLP waveform update is replaced by linear ZF precoding on each subcarrier. The multiuser communication channel matrix on the $n$-th subcarrier is
$\mathbf H_n(\mathbf d_T)
=
[\mathbf h_{n,1}(\mathbf d_T),\ldots,
\mathbf h_{n,K}(\mathbf d_T)]^H
\in\mathbb C^{K\times N_t}$,
and the corresponding ZF precoding matrix is given by\vspace{-0.1cm}
\begin{equation}
\mathbf V_n(\mathbf d_T)
=
\mathbf H_n^H(\mathbf d_T)
\left[
\mathbf H_n(\mathbf d_T)
\mathbf H_n^H(\mathbf d_T)
\right]^{-1},
\label{eq:zf_precoder}
\end{equation}
which satisfies $\mathbf H_n(\mathbf d_T)\mathbf V_n(\mathbf d_T)=\mathbf I_K$ and thereby eliminates multiuser spatial interference on the same subcarrier. All OFDM symbols on the $n$-th subcarrier share the same precoding matrix. Therefore, the ZF transmit waveform on the $(n,m)$-th RE is expressed as\vspace{-0.1cm}
\begin{equation}
\!\!\mathbf x_{n,m}^{\mathrm{ZF}}
\!\!=\!\!
\sqrt{
\frac{P_{\max}}{
\displaystyle
\sum_{m'=0}^{N_s-1}\!
\sum_{n'=0}^{N_c-1}\!
\left\|\!
\mathbf V_{n'}\!(\mathbf d_T)\mathbf s_{n',m'}
\!\right\|_2^2
}}\!
\mathbf V_n(\mathbf d_T)\mathbf s_{n,m}.
\label{eq:zf_transmit_waveform}
\end{equation}
\end{itemize}

\begin{figure}[t]
\centering
\includegraphics[width=0.9\linewidth]{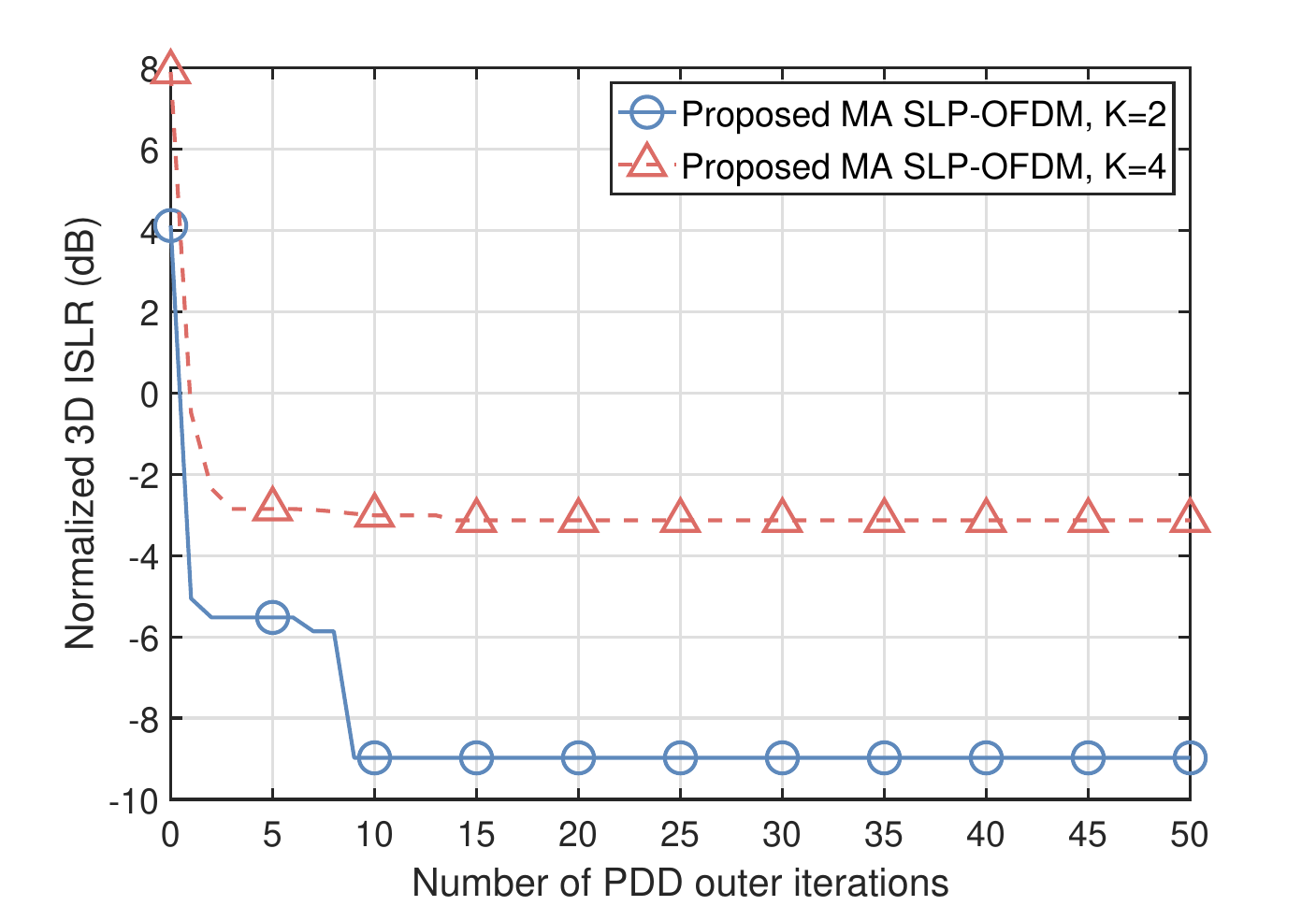}
\caption{Convergence behavior of the proposed algorithm.}
\label{fig:convergence}
\vspace{-0.2cm}
\end{figure}

Fig.~\ref{fig:convergence} illustrates the convergence behavior of the proposed PDD-AO algorithm for different numbers of users. The normalized 3D ISLR decreases rapidly during the first few outer iterations and then converges to a stable value, indicating reliable numerical convergence of the proposed algorithm. Moreover, when the number of users increases from $2$ to $4$, the converged 3D ISLR increases by approximately $5.84$~dB and the convergence becomes slower. This is mainly because the additional CI constraints reduce the feasible design space of the joint optimization problem.

\begin{figure}[t]
\centering
\subfigbottomskip=-1pt
\subfigcapskip=-5pt
\subfigure[Radar-only]
{\label{subfig:ar_radar_only}\includegraphics[width=0.32\linewidth]{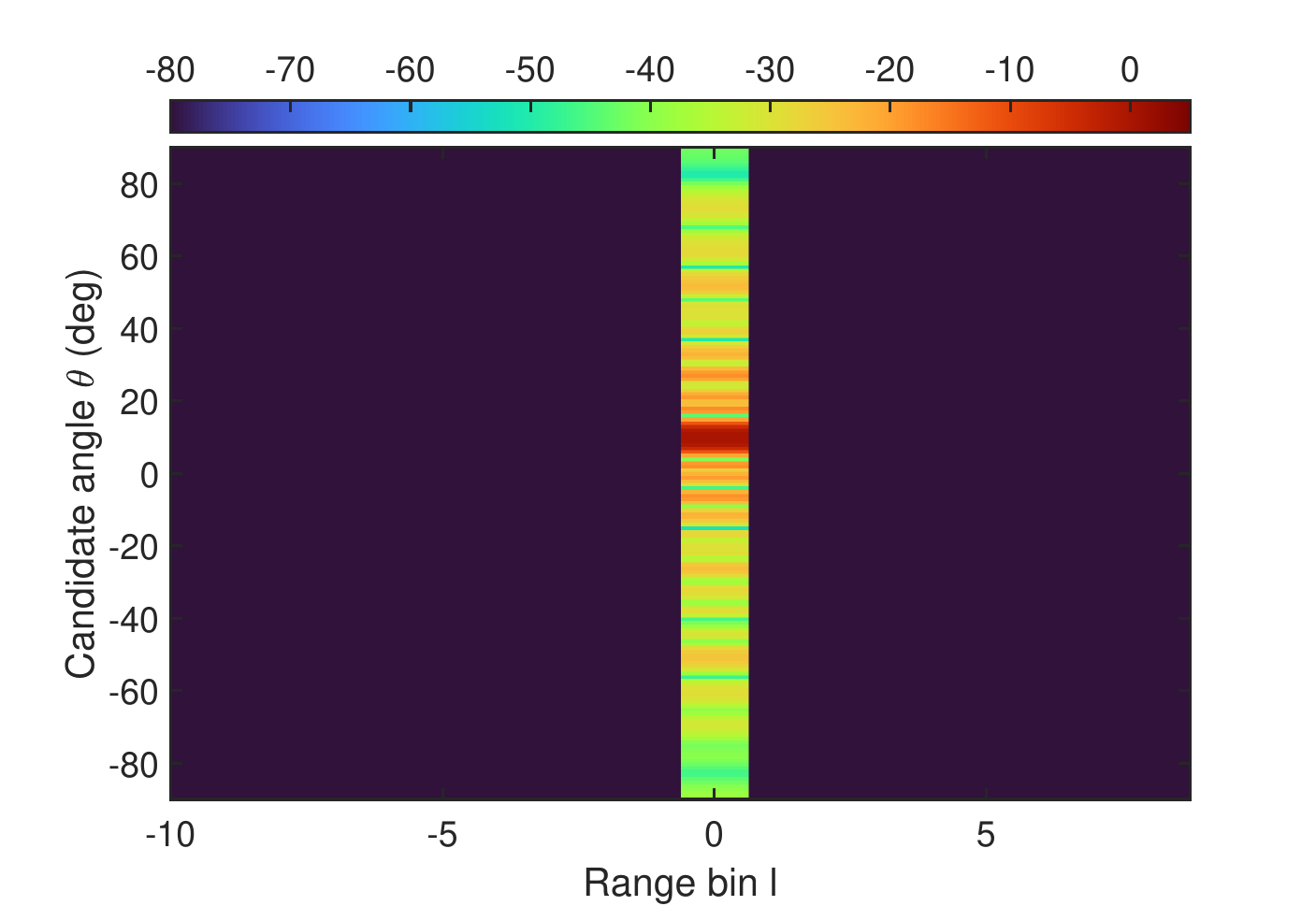}}
\subfigure[MA (Proposed)]
{\label{subfig:ar_proposed}\includegraphics[width=0.32\linewidth]{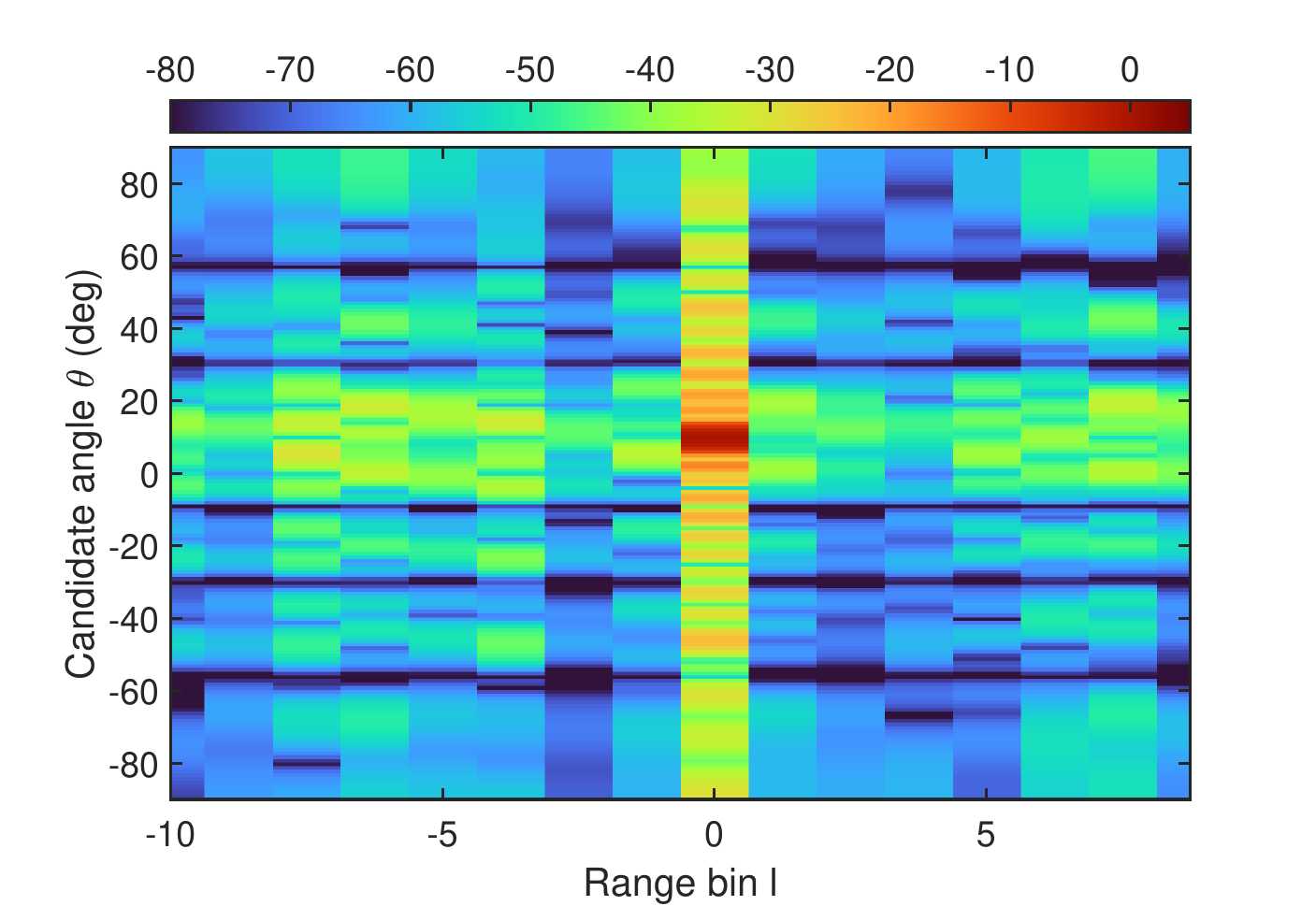}}
\subfigure[Compact FPA]
{\label{subfig:ar_fpa_dense}\includegraphics[width=0.32\linewidth]{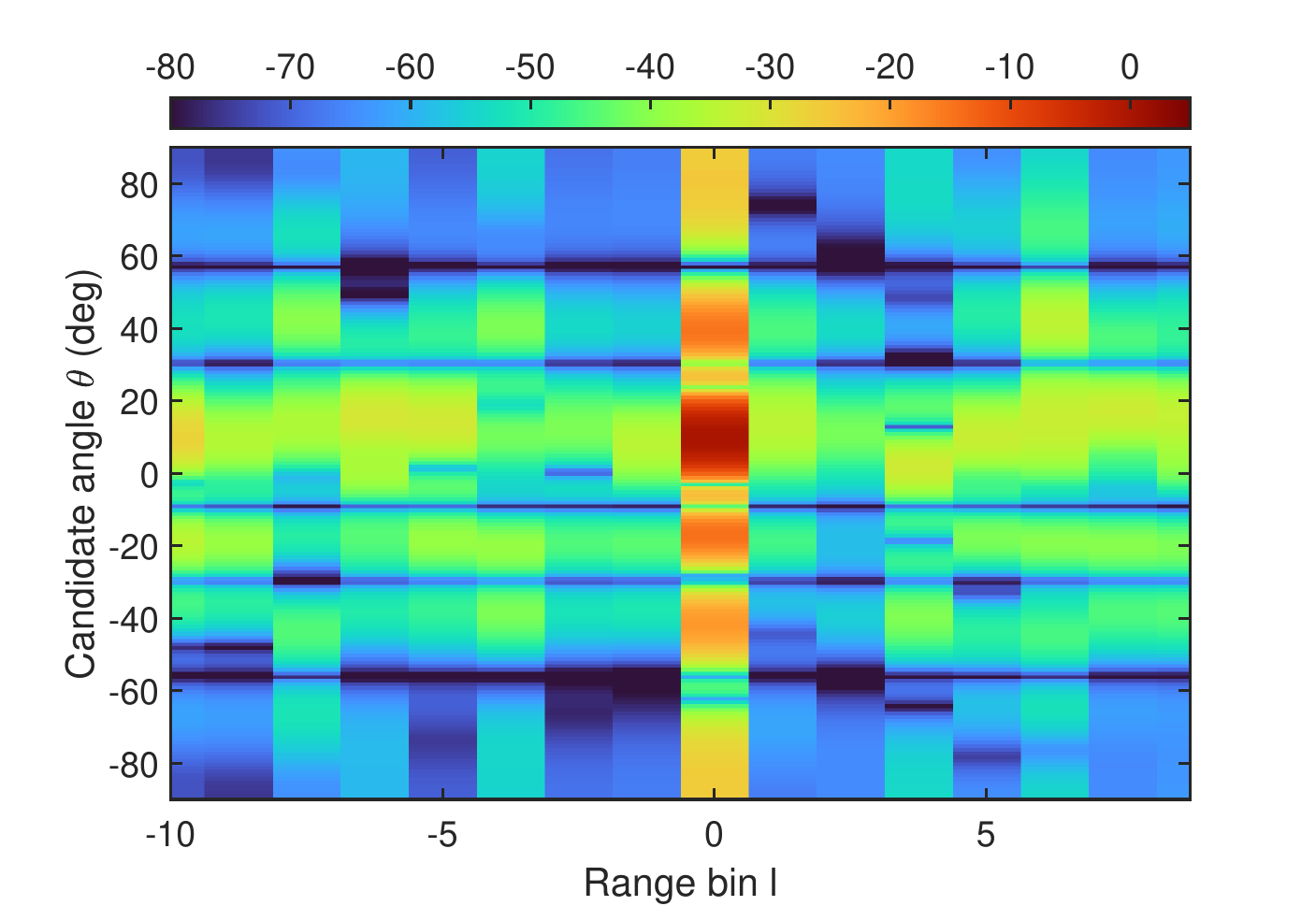}}
\\
\subfigure[Random FPA]
{\label{subfig:ar_fpa_random}\includegraphics[width=0.32\linewidth]{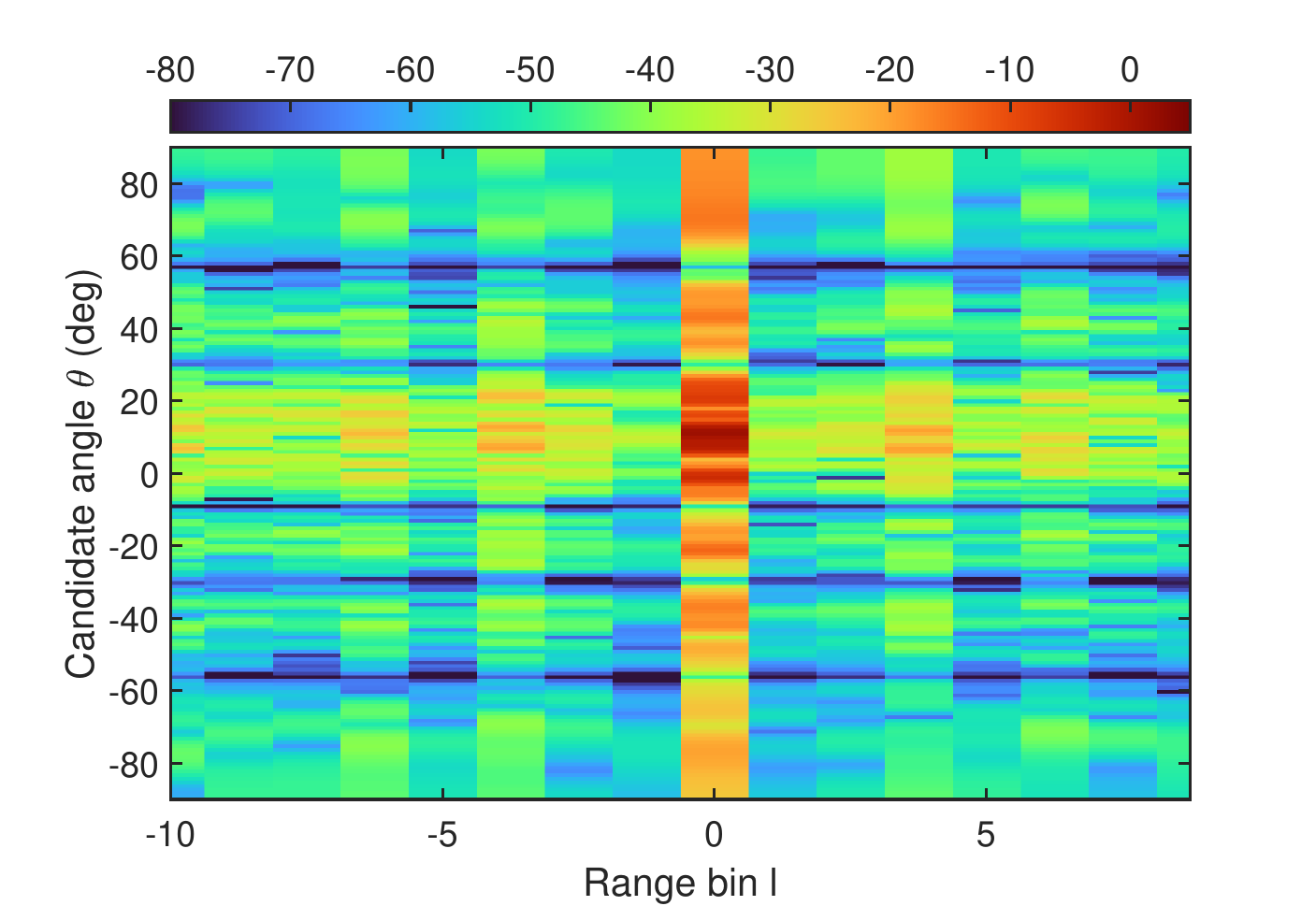}}
\subfigure[Sparse FPA]
{\label{subfig:ar_fpa_sparse}\includegraphics[width=0.32\linewidth]{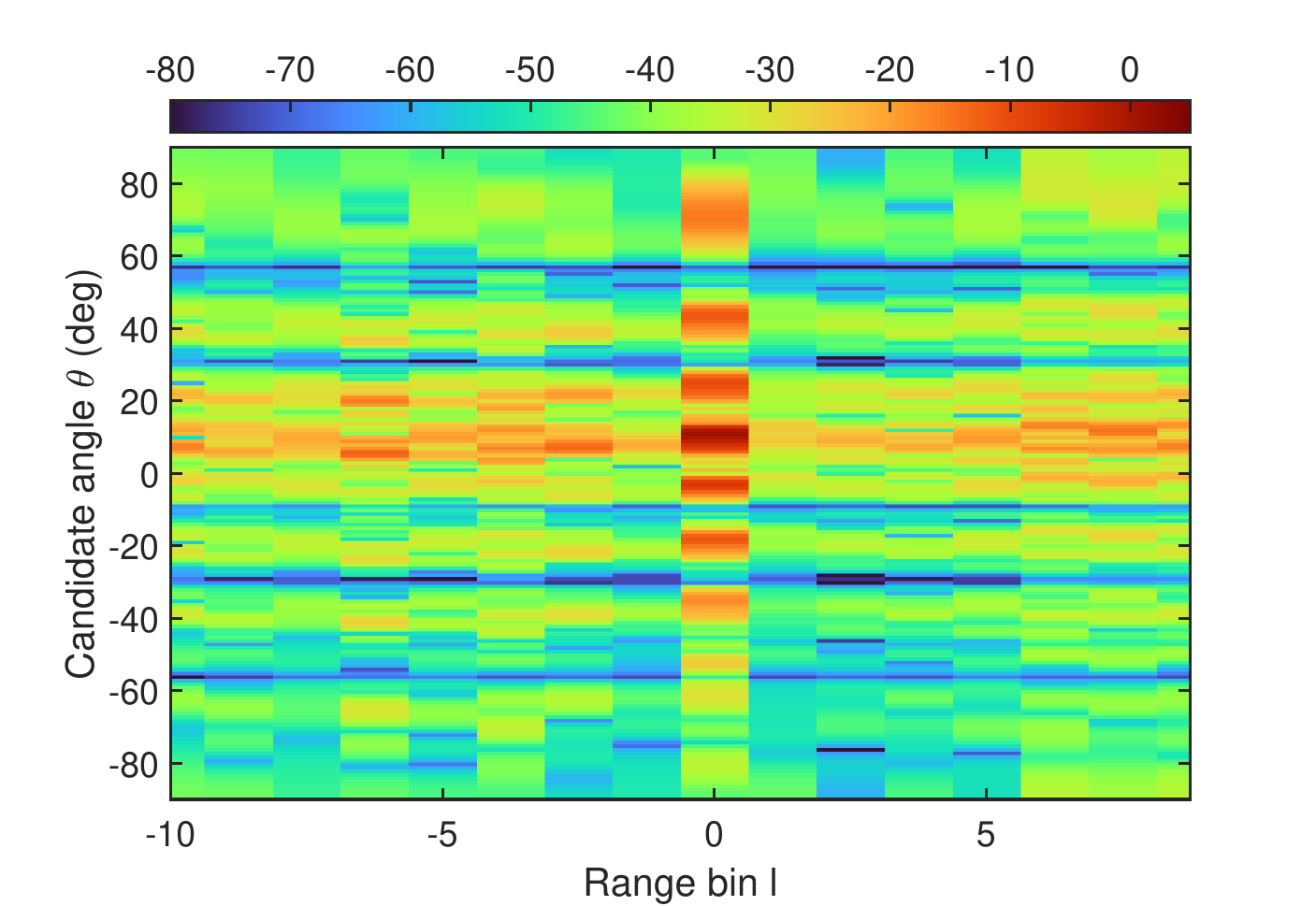}}
\subfigure[Comm-only]
{\label{subfig:ar_comm_only}\includegraphics[width=0.32\linewidth]{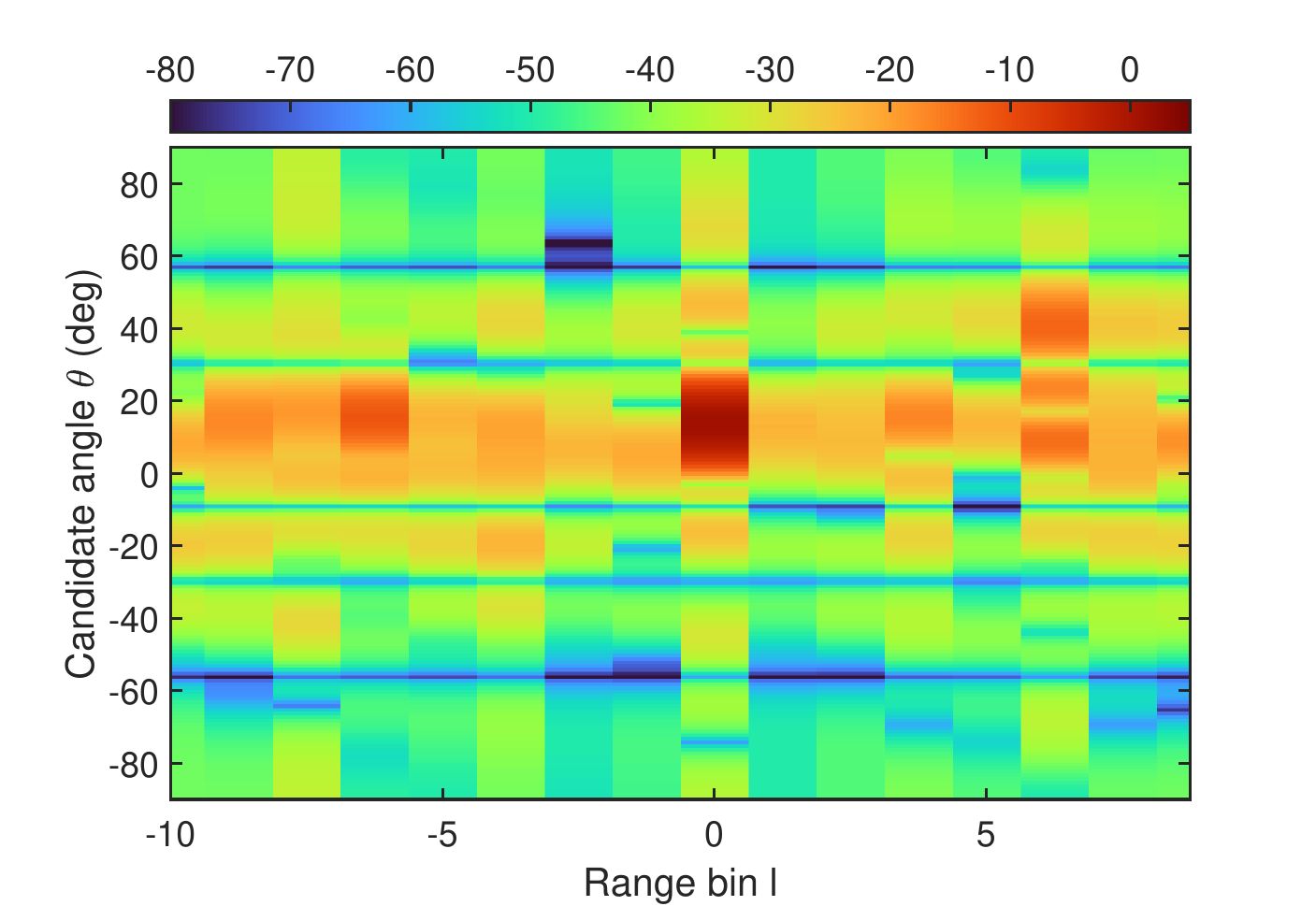}}
\\
\subfigure[Radar-only]
{\label{subfig:ad_radar_only}\includegraphics[width=0.32\linewidth]{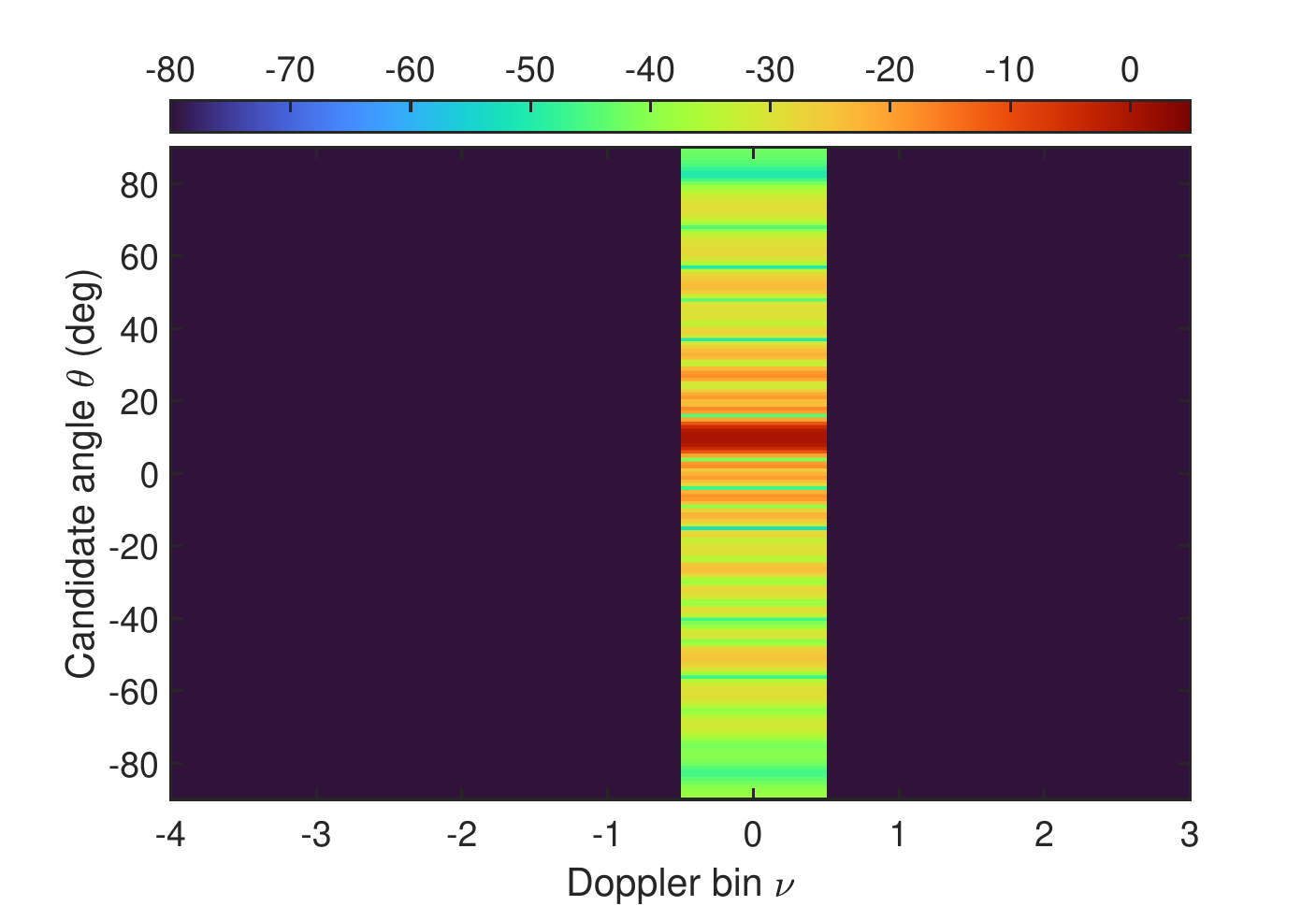}}
\subfigure[MA (Proposed)]
{\label{subfig:ad_proposed}\includegraphics[width=0.32\linewidth]{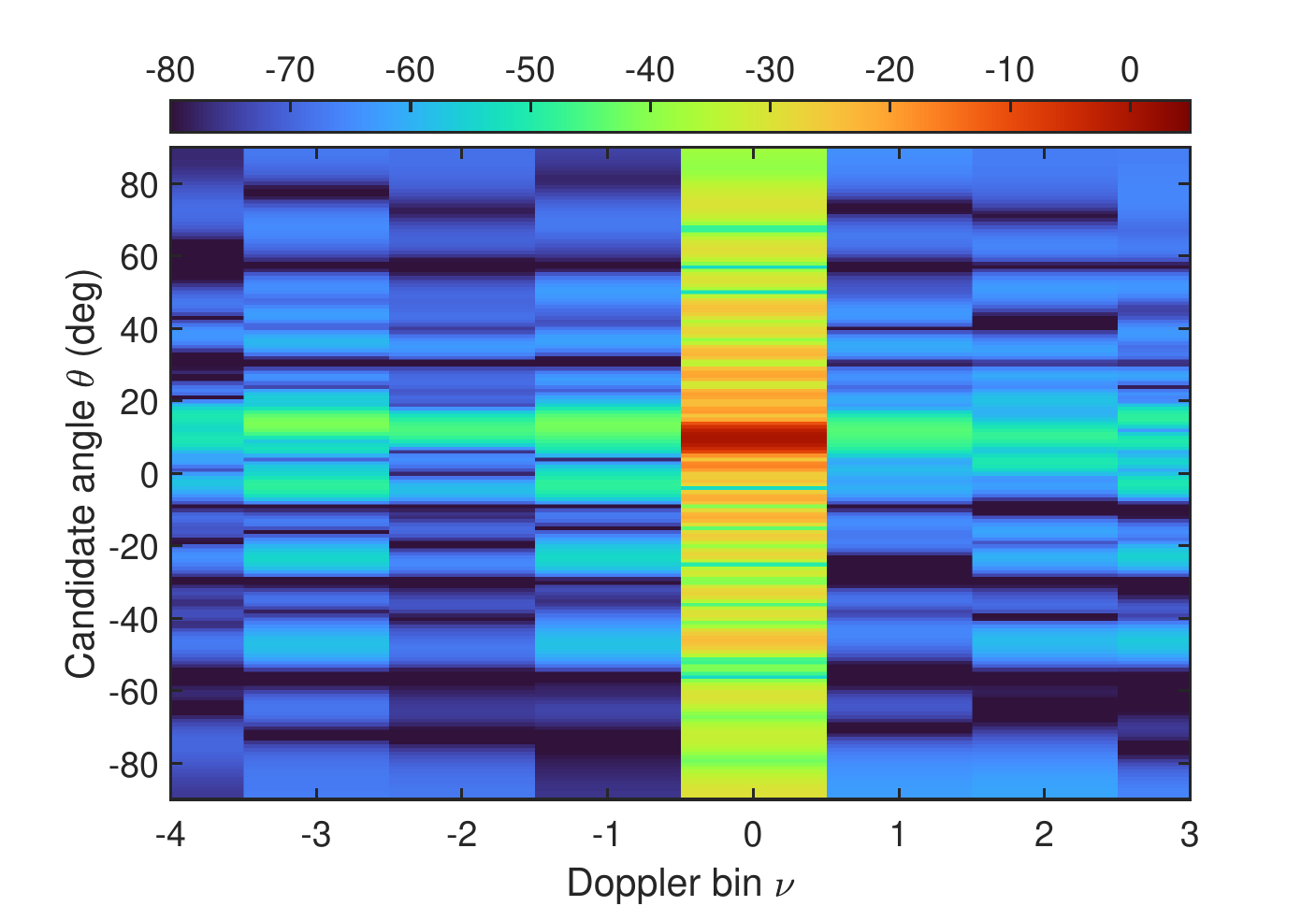}}
\subfigure[Compact FPA]
{\label{subfig:ad_fpa_dense}\includegraphics[width=0.32\linewidth]{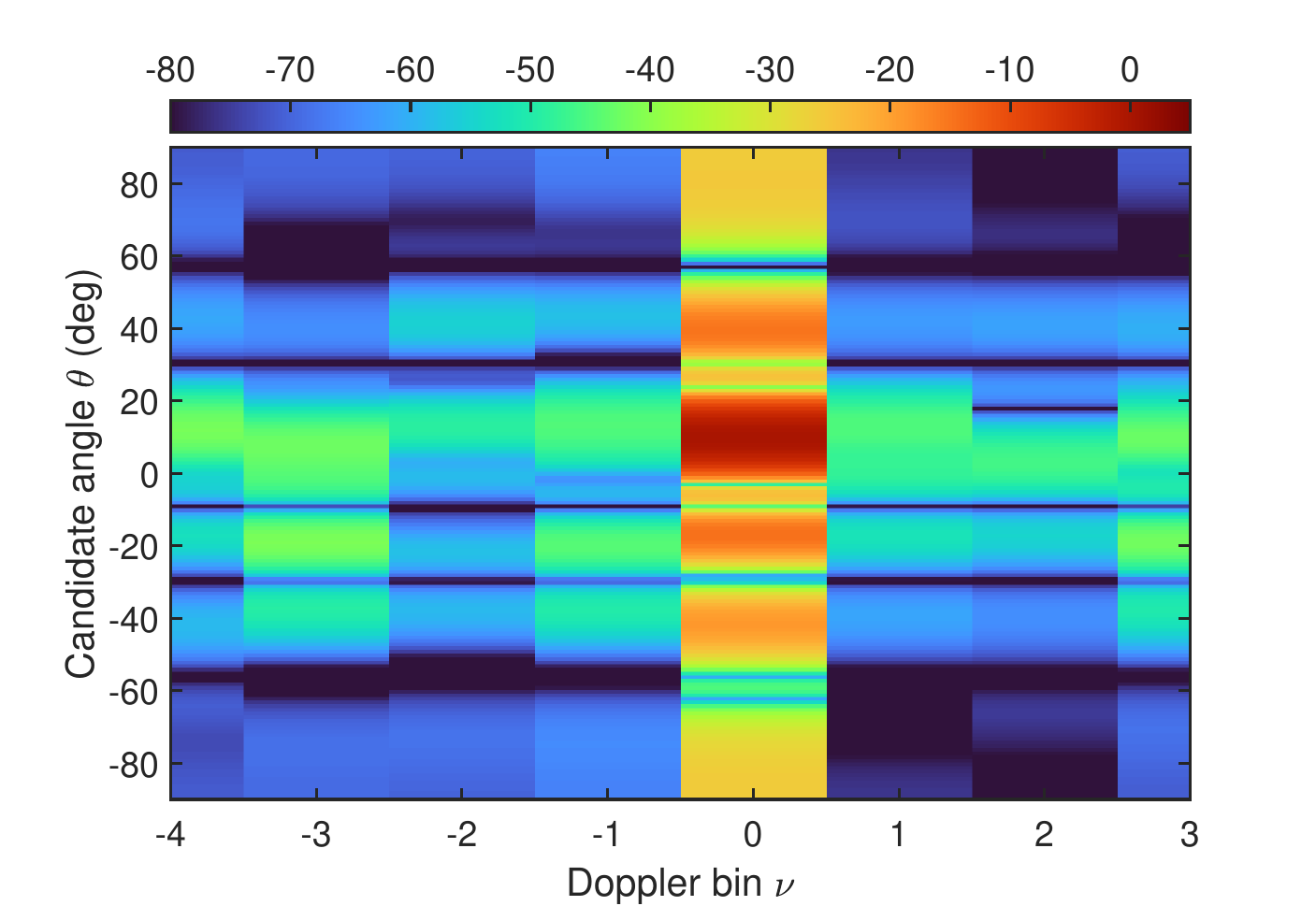}}
\\
\subfigure[Random FPA]
{\label{subfig:ad_fpa_random}\includegraphics[width=0.32\linewidth]{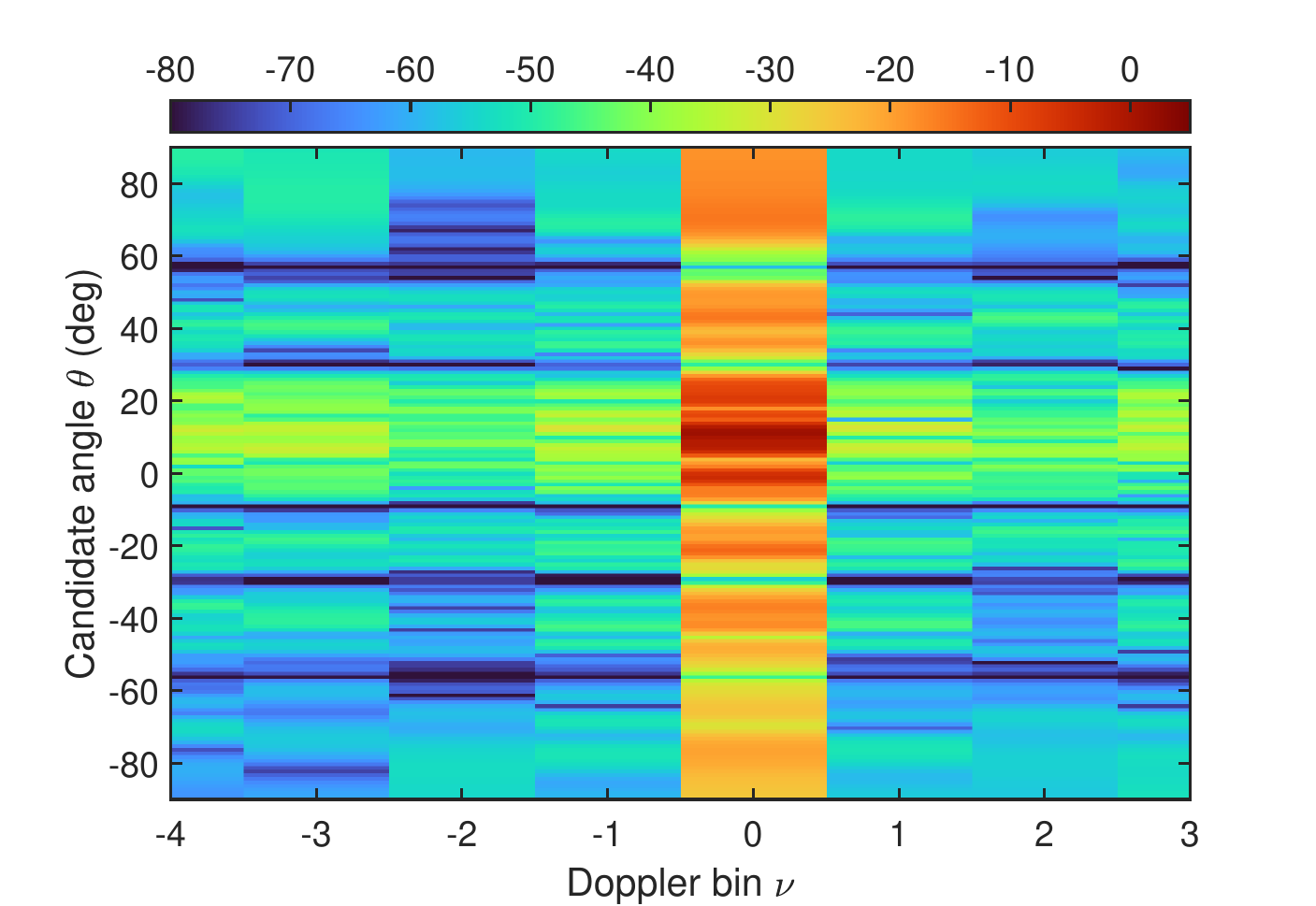}}
\subfigure[Sparse FPA]
{\label{subfig:ad_fpa_sparse}\includegraphics[width=0.32\linewidth]{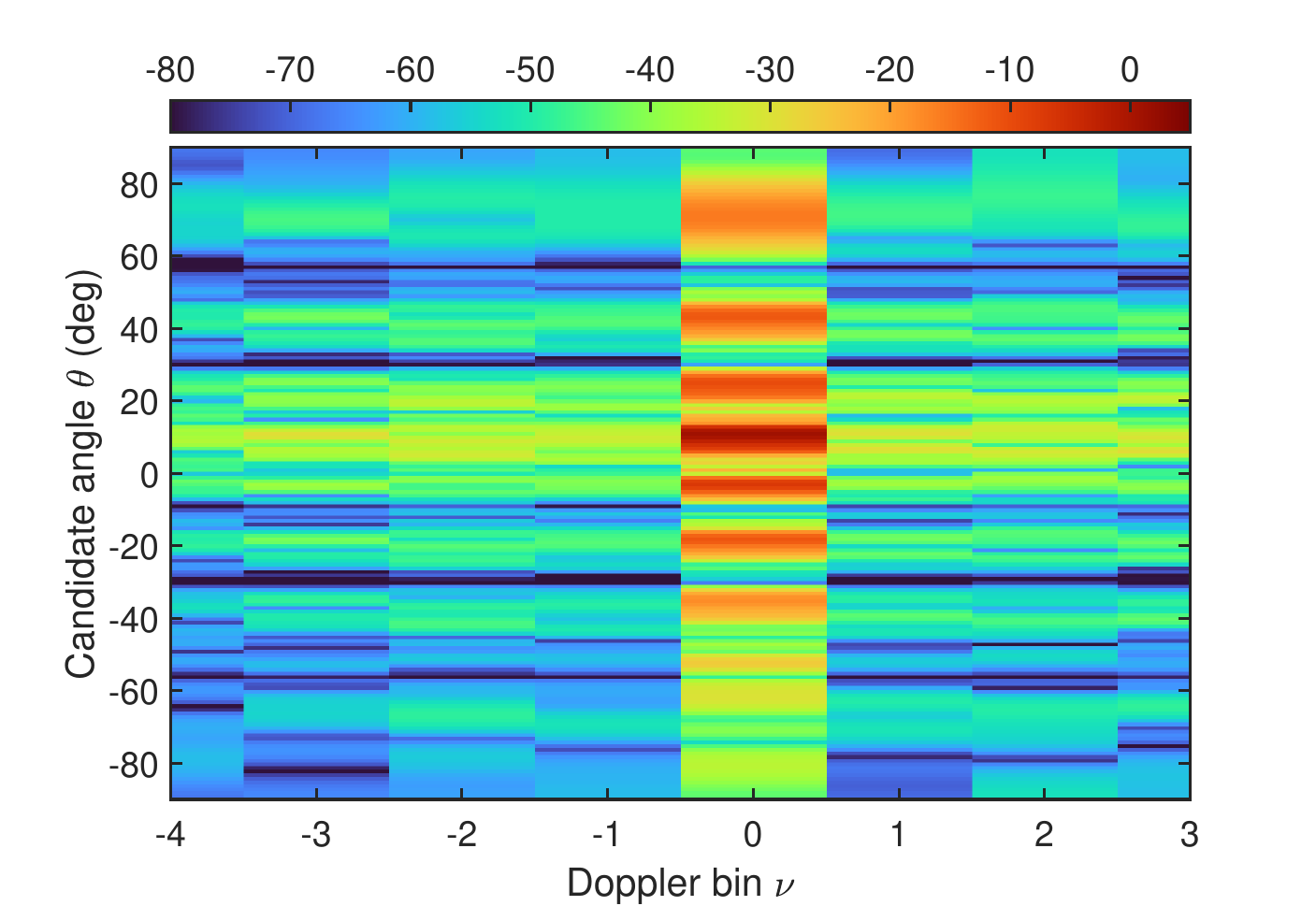}}
\subfigure[Comm-only]
{\label{subfig:ad_comm_only}\includegraphics[width=0.32\linewidth]{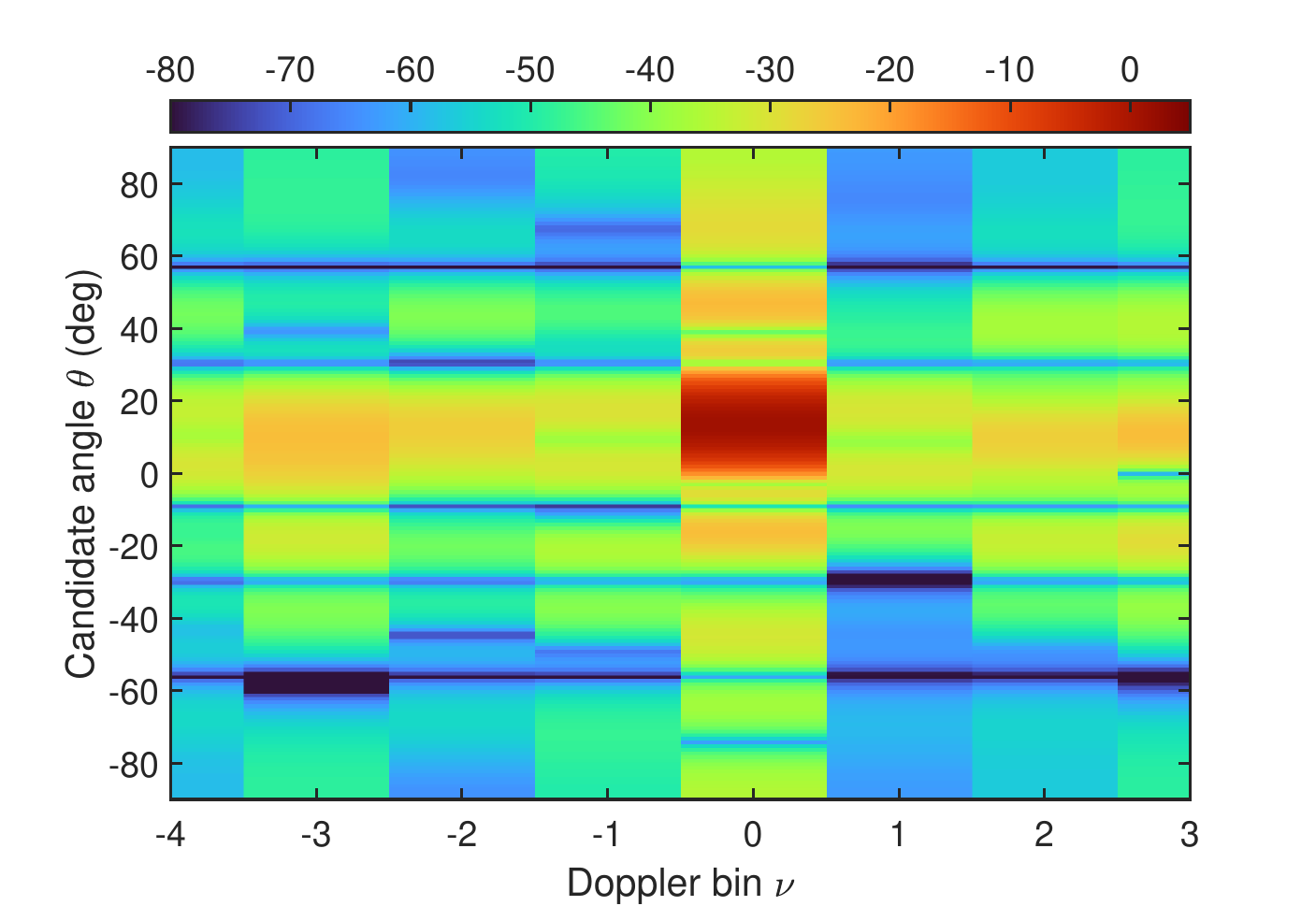}}
\caption{Normalized 3D ambiguity function slices for different schemes: (a)--(f) angle-range slices at zero Doppler and (g)--(l) angle--Doppler slices at zero delay.}
\label{fig:ambiguity_slices}
\end{figure}

Fig.~\ref{fig:ambiguity_slices} compares the normalized 3D ambiguity function slices of different schemes. The Comm-only scheme exhibits pronounced sidelobes at nonzero range and Doppler cells since its ambiguity sidelobes are not explicitly optimized. The three FPA SLP-OFDM schemes effectively suppress the range-Doppler sidelobes through per RE waveform design, while their angular responses remain constrained by the fixed array geometries. In particular, Compact FPA exhibits a broader angular response because of its smaller array aperture, whereas Random FPA and Sparse FPA exploit larger apertures but still exhibit noticeable angular sidelobes. In contrast, the proposed MA scheme produces a more concentrated main response around the target angle and zero range-Doppler cell while substantially reducing ambiguity leakage at non target angles and nonzero range-Doppler cells. At the matched angle $\theta=\theta_0=10^\circ$, its peak range and Doppler sidelobe levels are approximately $10.64$~dB and $3.44$~dB lower, respectively, than those of Compact FPA. Moreover, the Radar-only scheme exhibits nearly zero response at discrete nonzero range-Doppler cells, consistent with the rank-one equal RE energy structure derived in Section~\ref{subsec:radar_only_case}.

\begin{figure}[t]
\centering
\includegraphics[width=0.9\linewidth]{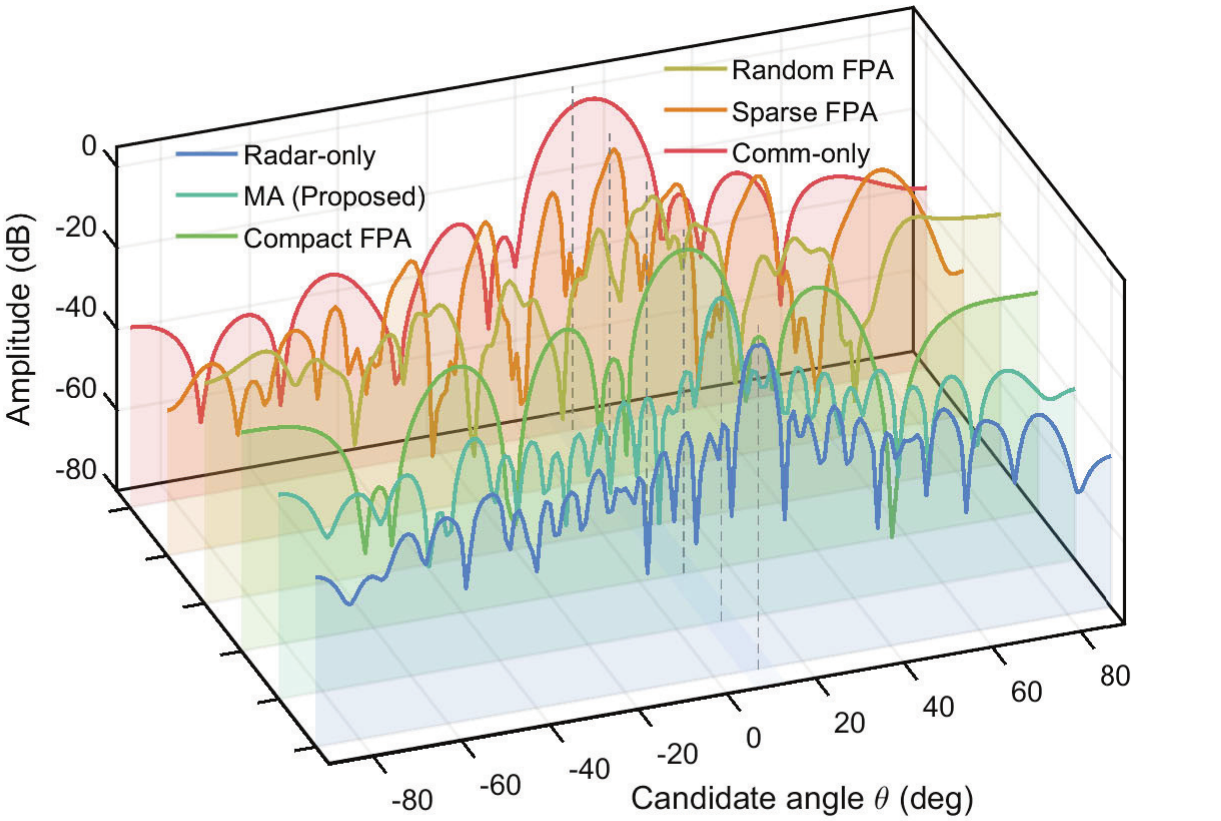}
\caption{Normalized angular ambiguity functions at zero delay and zero Doppler for different schemes.}
\label{fig:angular_ambiguity}
\vspace{-0.2cm}
\end{figure}

Fig.~\ref{fig:angular_ambiguity} further compares the angular slices of the normalized 3D ambiguity functions at zero delay and zero Doppler. Compact FPA exhibits a broad angular mainlobe because of its limited array aperture, whereas Random FPA and Sparse FPA, despite their larger apertures, still exhibit pronounced angular sidelobes. By jointly optimizing the MA positions and SLP waveform, the proposed MA scheme substantially suppresses angular leakage outside the mainlobe region. Its angular ISLR approaches that of Radar-only, whereas Comm-only exhibits the largest angular leakage, indicating that flexible antenna positioning enables more effective use of the available aperture for angular sidelobe suppression.

\begin{figure}[t]
\centering
\includegraphics[width=0.9\linewidth]{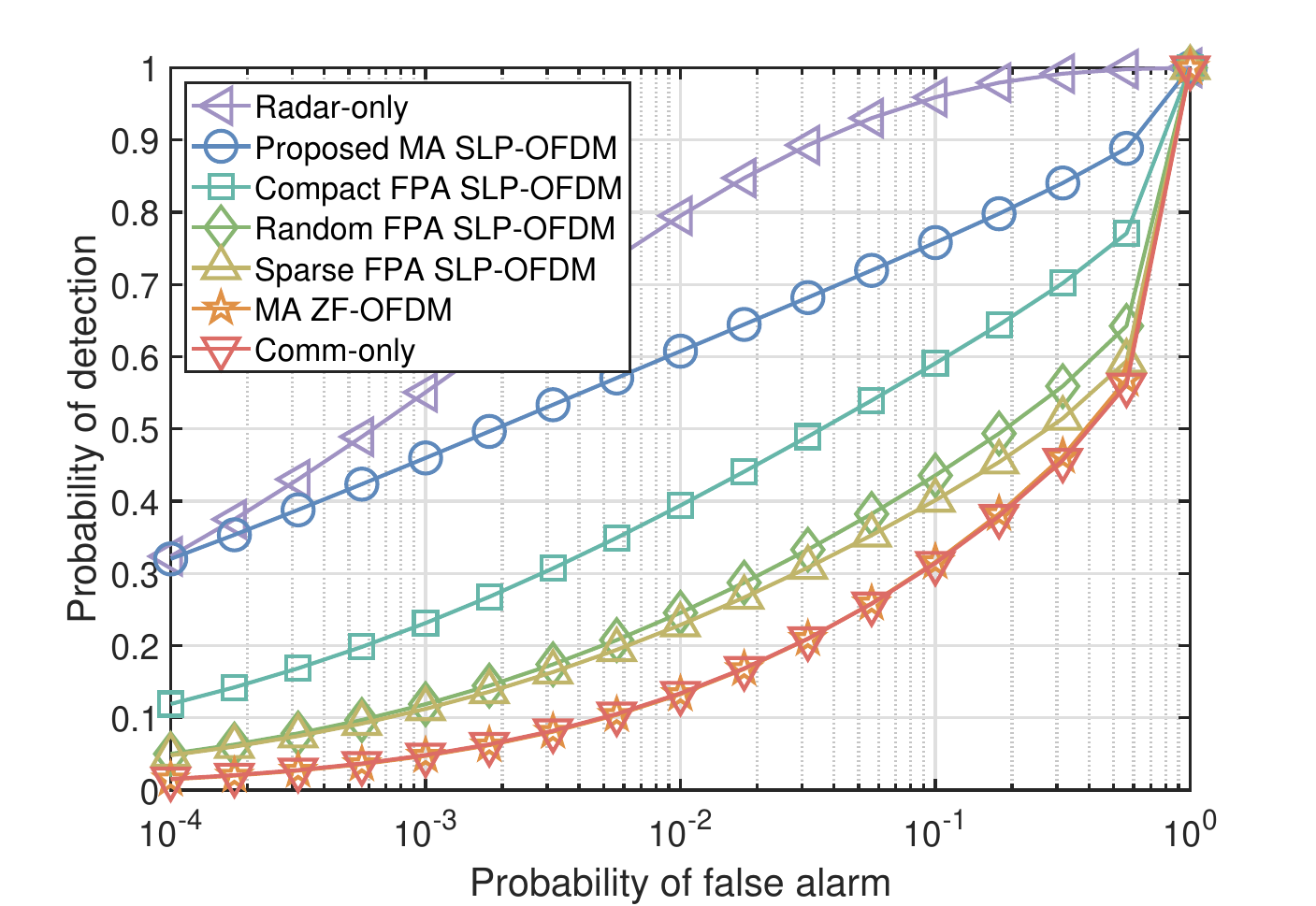}
\caption{Weak target detection ROC curves in the presence of a strong target, where the strong and weak target RCS values are $20$~dBsm and $-5$~dBsm, respectively, with $N_c=16$ and $N_s=8$.}
\label{fig:detection_roc}
\vspace{-0.2cm}
\end{figure}

Fig.~\ref{fig:detection_roc} compares the weak target detection performance of different schemes in the presence of a strong target. The proposed MA SLP-OFDM scheme achieves consistently higher detection probabilities than the three FPA SLP-OFDM schemes, MA ZF-OFDM, and Comm-only. At $P_{\mathrm{fa}}=10^{-3}$, the detection probability of the proposed scheme is approximately $0.460$, compared with $0.231$, $0.119$, and $0.112$ for Compact FPA SLP-OFDM, Random FPA SLP-OFDM, and Sparse FPA SLP-OFDM, respectively, while those of MA ZF-OFDM and Comm-only are both below $0.05$. This gain results from suppressing the sidelobe leakage of the strong target at nonzero range-Doppler cells, which mitigates the masking of the weak target and improves its detection probability. Radar-only achieves the highest detection probability, while the proposed scheme retains receiver operating characteristic (ROC) performance close to Radar-only while satisfying the multiuser communication requirements.

\begin{figure}[t]
\centering
\includegraphics[width=0.9\linewidth]{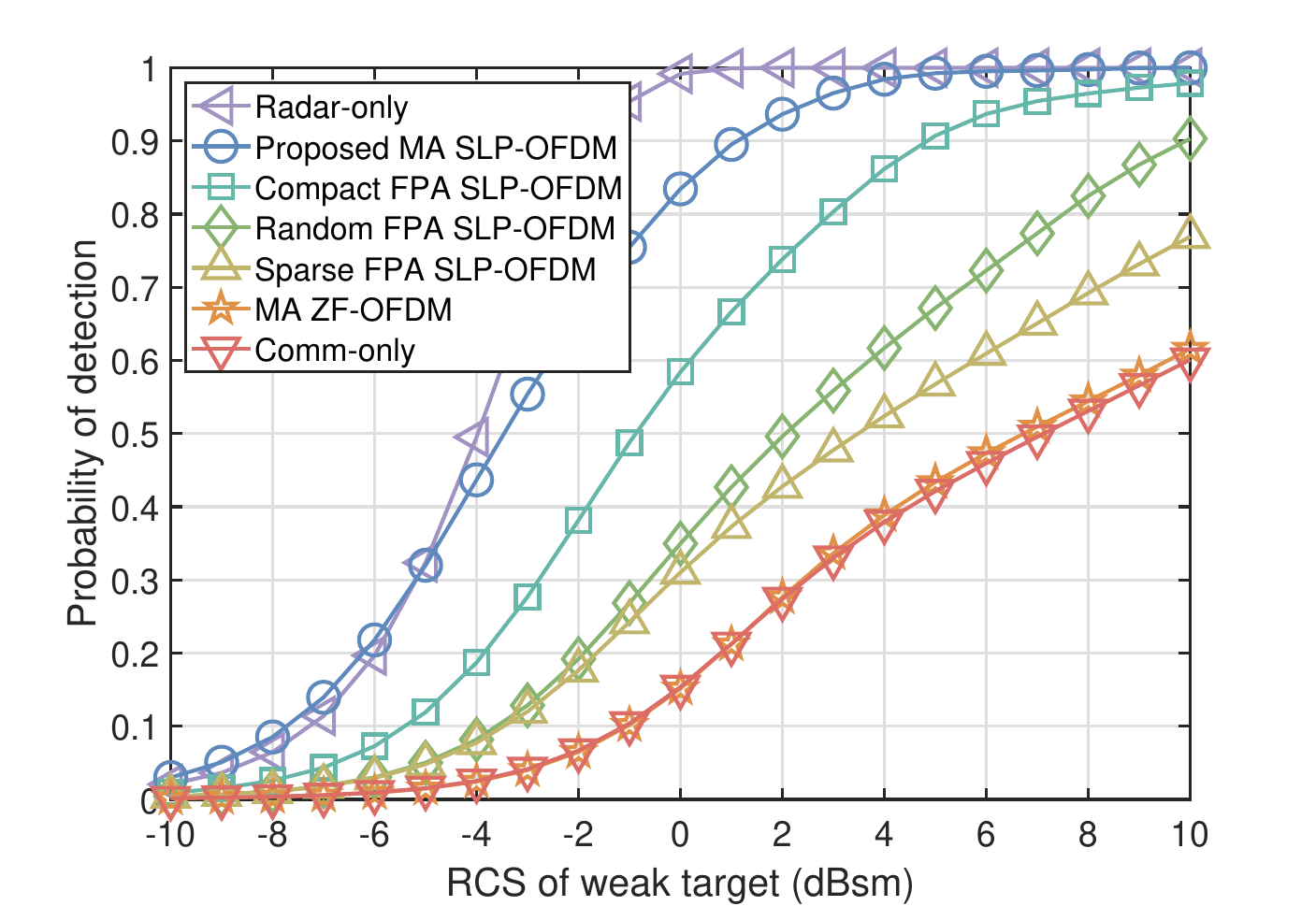}
\caption{Weak target detection probability versus RCS for different schemes, where $P_{\mathrm{fa}}=10^{-4}$, the strong target RCS is $20$~dBsm, and $N_c=16$ and $N_s=8$.}
\label{fig:detection_probability_rcs}
\end{figure}

Fig.~\ref{fig:detection_probability_rcs} further investigates the impact of the weak target radar cross section (RCS) on the detection performance. As the RCS increases, the detection probabilities of all schemes improve, while the proposed MA SLP-OFDM scheme maintains higher detection probabilities in the low RCS region, consistent with its reduced sidelobe leakage from the strong target. As the weak target becomes stronger, the detection probabilities gradually approach saturation and the performance gaps between the schemes diminish.

\begin{figure}[t]
\centering
\includegraphics[width=0.9\linewidth]{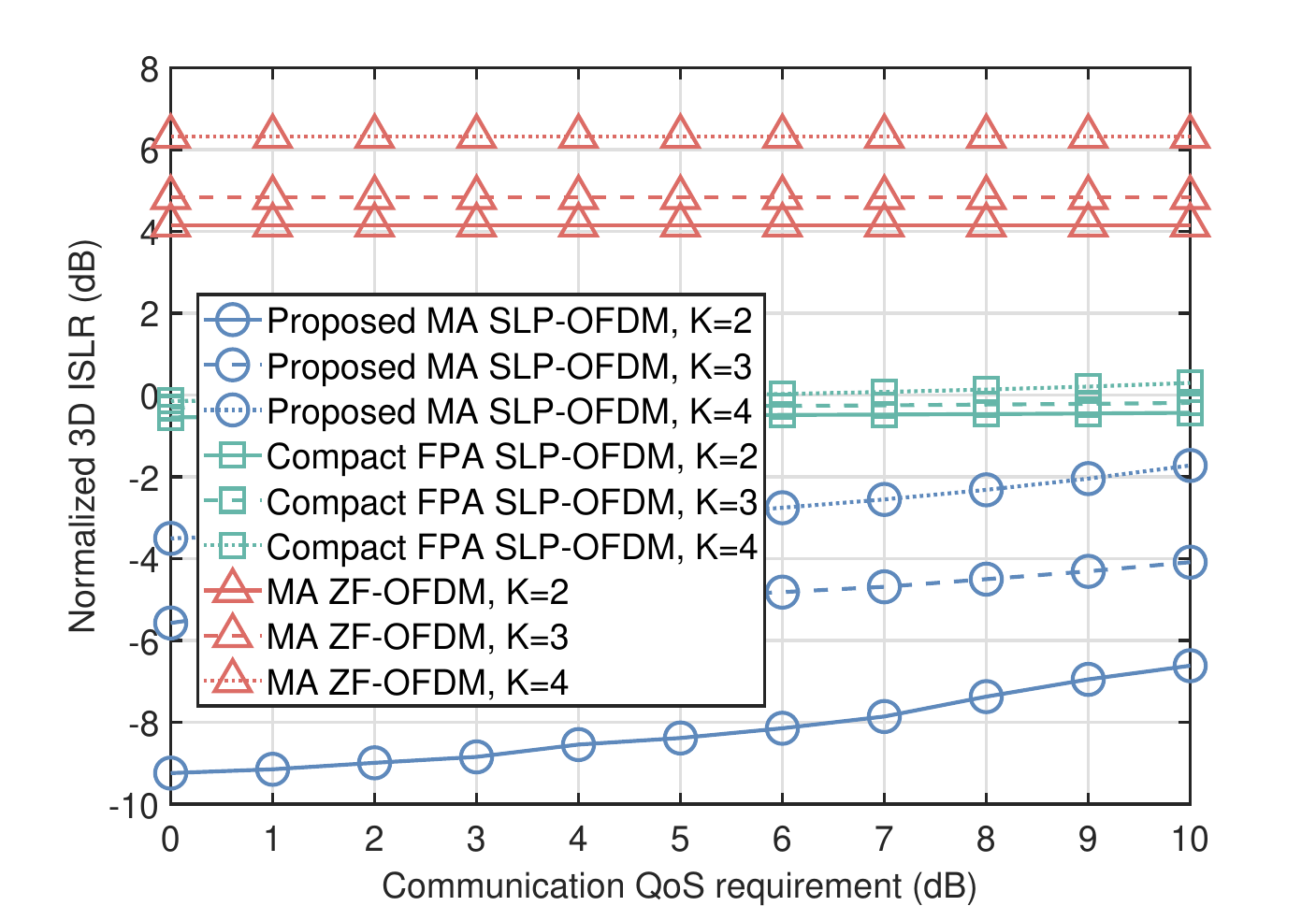}
\caption{Communication and sensing performance tradeoff for different numbers of users.}
\label{fig:isac_tradeoff}
\vspace{-0.2cm}
\end{figure}

Fig.~\ref{fig:isac_tradeoff} shows the communication and sensing performance tradeoff for different numbers of users. As the QoS threshold or the number of users increases, the CI constraints become more stringent and the 3D ISLRs of the SLP schemes generally increase. The proposed MA SLP-OFDM scheme consistently achieves lower 3D ISLRs than the three FPA SLP-OFDM baselines, indicating that reconfigurable array geometry provides additional flexibility for suppressing 3D ambiguity sidelobes under communication constraints. The proposed scheme also achieves lower 3D ISLRs than MA ZF-OFDM overall, showing that SLP can exploit instantaneous symbol information to satisfy the CI constraints more flexibly while preserving more waveform design freedom for sensing sidelobe suppression. The 3D ISLR of MA ZF-OFDM varies only slightly with the QoS threshold because changing the QoS requirement mainly affects the common amplitude scaling of the ZF waveform, whereas the 3D ISLR is invariant to such overall scaling.

\begin{figure}[t]
\centering
\includegraphics[width=0.9\linewidth]{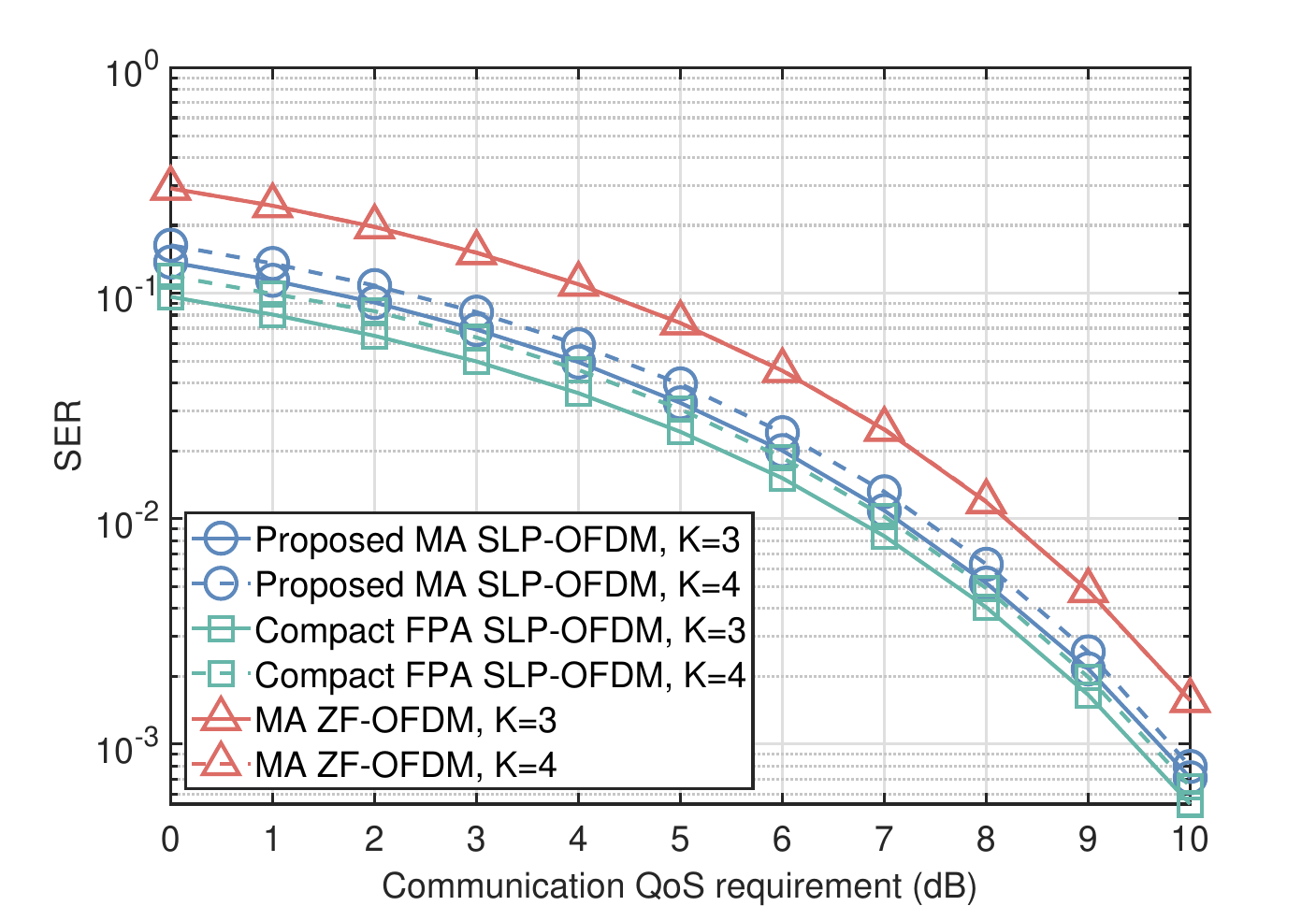}
\caption{SER versus communication QoS threshold for different numbers of users.}
\label{fig:ser_qos}
\end{figure}

Fig.~\ref{fig:ser_qos} compares the SER performance of the proposed MA SLP-OFDM, Compact FPA SLP-OFDM, and MA ZF-OFDM schemes for different numbers of users. As the communication QoS threshold increases, the SERs of all schemes decrease due to the enlarged CI margins and improved noise robustness. Both SLP schemes outperform MA ZF-OFDM, confirming the benefit of constructive interference exploitation. Compact FPA SLP-OFDM achieves lower SERs than the proposed scheme since the latter minimizes the 3D ISLR with communication QoS imposed only as a constraint, allowing the MA spatial DoFs to be exploited primarily for ambiguity sidelobe suppression rather than SER minimization.
%Fig.~\ref{fig:ser_qos} compares the symbol error rate (SER) performance of the proposed MA SLP-OFDM and MA ZF-OFDM schemes for different numbers of users. As the communication QoS threshold increases, the SERs of both schemes decrease because a higher QoS requirement increases the distance between the noiseless received symbols and their decision boundaries, thereby improving robustness against noise perturbations. For both user settings, the proposed scheme achieves lower SERs than MA ZF-OFDM, demonstrating the benefit of exploiting constructive multiuser interference through SLP for reliable symbol detection.

\begin{figure}[t]
\centering
\includegraphics[width=0.9\linewidth]{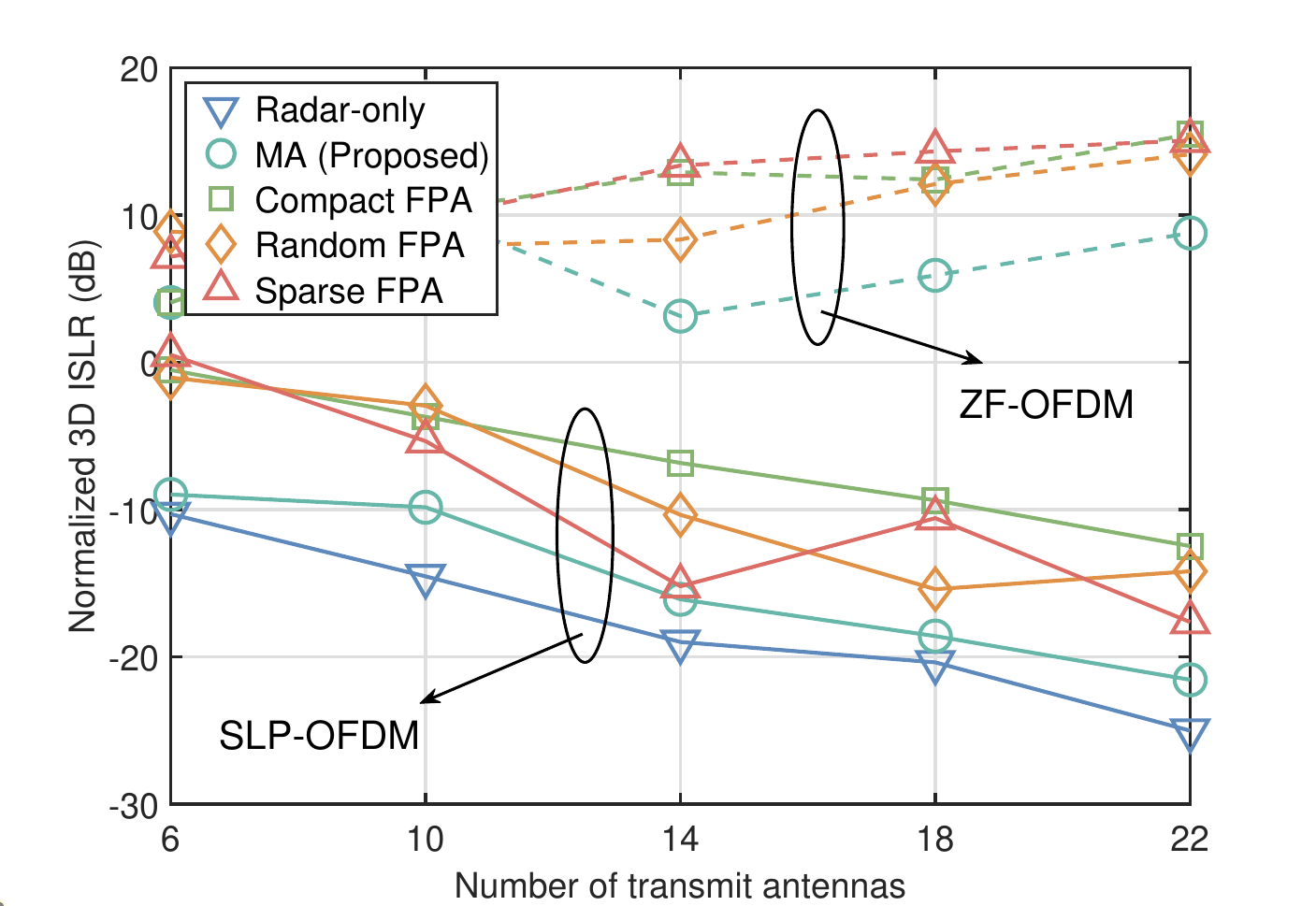}
\caption{3D ISLR versus the number of transmit antennas for different schemes.}
\label{fig:islr_number_antennas}
\vspace{-0.2cm}
\end{figure}

Fig.~\ref{fig:islr_number_antennas} compares the 3D ISLRs of different schemes versus the number of transmit antennas. As $N_t$ increases, the 3D ISLR of the proposed MA SLP-OFDM scheme steadily decreases, indicating that additional transmit antennas provide richer spatial DoF for the joint optimization of the array geometry and transmit waveforms on individual REs. Across the considered range of $N_t$, the proposed scheme maintains lower 3D ISLRs than the three FPA SLP-OFDM schemes, demonstrating the additional sensing gain enabled by MA position optimization beyond waveform design alone. The nonmonotonic behavior of Random FPA SLP-OFDM and Sparse FPA SLP-OFDM suggests that increasing $N_t$ does not automatically translate into improved 3D sidelobe suppression when the array geometry is prescribed rather than optimized. Notably, the 3D ISLR of ZF-OFDM increases slightly with $N_t$ because the right pseudoinverse precoder does not exploit the additional channel nullspace DoF for 3D ambiguity shaping. In contrast, SLP directly optimizes the transmit waveform on each RE and jointly exploits the additional spatial DoF with the MA positions, leading to a clearer reduction in 3D ISLR as $N_t$ increases.

\vspace{-0.2cm}

\section{Conclusion}\label{section:conclusion}
This paper proposed an MA-empowered MIMO-OFDM ISAC framework and revealed the intrinsic coupling between the MA array geometry and the MIMO-OFDM waveform in shaping the 3D ambiguity function. This coupling can be characterized by the waveform rank. Specifically, a rank-one waveform yields a 3D ambiguity function that is separable between the angular dimension and the range-Doppler plane, whereas a higher-rank waveform allows the range-Doppler ambiguity function to vary with the candidate angle, enabling the MA positions to further shape the 3D ambiguity structure. For the radar-only problem considered herein, we proved that at least one rank-one globally optimal solution exists, leading to an equivalent low-dimensional formulation of the original waveform design problem. Based on these structural insights, the joint optimization of the MA positions and SLP waveform using the 3D ISLR effectively suppresses 3D ambiguity sidelobes and improves weak target detection in the presence of strong target interference.

\vspace{-0.2cm}

\renewcommand{\IEEEbibitemsep}{-0.4pt plus 0.2pt}
\bibliographystyle{IEEEtran}
\bibliography{references}

\begin{thebibliography}{32}

\bibitem{tan2026joint}
S. Tan, F. Luo, Y. Gang, and L. Fang, ``Joint transmit and receive beamforming design for cell-free ISAC MIMO systems,'' in \textit{Proc. IEEE Wireless Commun. Netw. Conf. (WCNC) Workshops}, 2026.

\bibitem{liu2025itu}
R. Liu, L. Zhang, R. Y. -N. Li, and M. D. Renzo, “The ITU vision and framework for 6G: Scenarios, capabilities, and enablers,” \textit{IEEE Veh. Technol. Mag.}, vol. 20, no. 2, pp. 114--122, Jun. 2025.

\bibitem{mao2022waveform}
T. Mao, J. Chen, Q. Wang, C. Han, Z. Wang, and G. K. Karagiannidis, “Waveform design for joint sensing and communications in millimeter-wave and low terahertz bands,” \textit{IEEE Trans. Commun.}, vol. 70, no. 10, pp. 7023--7039, Oct. 2022.

\bibitem{zhang2026integrated}
D. Zhang et al., "Integrated sensing and communications over the years: An evolution perspective," \textit{IEEE Commun. Surveys Tuts.}, vol. 28, pp. 5014--5048, 2026. 

\bibitem{liu2026sensing}
F. Liu, Y.-F. Liu, Y. Cui, C. Masouros, J. Xu, T.-X. Han, S. Buzzi, Y. C. Eldar, and S. Jin, “Sensing with communication signals: From information theory to signal processing,” \textit{IEEE J. Sel. Areas Commun.}, vol. 44, pp. 1--30, 2026.

\bibitem{zhang2025afdm}
F. Zhang Z. Wang, T. Mao, T. Jiao, Y. Zhuo, M. Wen, W. Xiang, S. Chen, and G. K. Karagiannidis, “AFDM-enabled integrated sensing and communication: Theoretical framework and pilot design,” \textit{IEEE J. Sel. Areas Commun.}, vol. 44, pp. 310--324, 2026.

\bibitem{liu20246g}
R. Liu et al., “6G enabled advanced transportation systems,” \textit{IEEE Trans. Intell. Transp. Syst.}, vol. 25, no. 9, pp. 10564--10580, Sep. 2024.

\bibitem{gang2025uav}
Y. Gang, Y. Zhang, and X. Wang, “UAV-assisted full-duplex ISAC: Joint communication scheduling, beamforming, and trajectory optimization,” \textit{Digital Commun. Netw.}, vol. 11, no. 5, pp. 1628--1638, Oct. 2025.

\bibitem{li2026intelligent}
S. Li et al., “Intelligent metasurface-enabled integrated sensing and communication: Unified framework and key technologies,” \textit{IEEE Wireless Commun.}, vol. 33, no. 1, pp. 216--223, Feb. 2026.

\bibitem{zhang2025integrated}
Y. Zhang Y. Gang, P. Wu, G. Fan, W. Xu, B. Ai, and Q. Wu, “Integrated sensing, communication, and computation in SAGIN: Joint beamforming and resource allocation,” \textit{IEEE Trans. Cogn. Commun. Netw.}, vol. 11, no. 5, pp. 3128--3143, Oct. 2025.

\bibitem{li2025transmit}
S. Li, F. Zhang, T. Mao, R. Na, Z. Wang, and G. K. Karagiannidis, “Transmit beamforming design for ISAC with stacked intelligent metasurfaces,” \textit{IEEE Trans. Veh. Technol.}, vol. 74, no. 4, pp. 6767--6772, Apr. 2025.

\bibitem{liu2024joint}
S. Liu, M. Li, R. Liu, W. Wang, and Q. Liu, “Joint transmit beamforming and receive filter design for cooperative multi-static ISAC networks," \textit{IEEE Wireless Commun. Lett.}, vol. 13, no. 6, pp. 1700--1704, Jun. 2024.

\bibitem{zhang2026leo}
Y. Zhang, Y. Gang, T. Mao, R. Liu, M. Hua, and Q. Wu, “LEO satellite-enabled ISAC: Bistatic framework and beamforming design,” \textit{IEEE Trans. Veh. Technol.}, doi: 10.1109/TVT.2026.3686270.

\bibitem{yang2026cooperative}
Y. Yang et al., “Cooperative multi-static ISAC networks: A unified design framework for active and passive sensing," \textit{IEEE Trans. Wireless Commun.}, vol. 25, pp. 5401--5415, 2026.

\bibitem{huang2022coordinated}
Y. Huang, Y. Fang, X. Li, and J. Xu, “Coordinated power control for network integrated sensing and communication,” \textit{IEEE Trans. Veh. Technol.}, vol. 71, no. 12, pp. 13361--13365, Dec. 2022.

\bibitem{zhang2026multicell}
J. Zhang, C. Qi, S. Mao, and O. A. Dobre, “Multi-cell integrated sensing and communication: Cooperative passive sensing and resource allocation," \textit{IEEE Trans. Commun.}, vol. 74, pp. 10202--10215, 2026.

\bibitem{dou2024integrated}
C. Dou, N. Huang, Y. Wu, L. Qian, Z. Shi, and T. Q. S. Quek, “Integrated sensing and communication enabled multidevice multitarget cooperative sensing: A Fairness-aware design," \textit{IEEE Internet Things J.}, vol. 11, no. 17, pp. 29190--29201, Sep. 2024.

\bibitem{yang2025cooperative}
X. Yang, Z. Wei, J. Xu, H. Wu, and Z. Feng, “Cooperative sensing-assisted predictive beam tracking for MIMO-OFDM networked ISAC systems," \textit{IEEE Trans. Wireless Commun.}, vol. 24, no. 12, pp. 10660--10674, Dec. 2025.

\bibitem{liu2026sensingwith}
R. Liu, L. Zhang, Y. Gang, T. Mao, Q. Wu, and A. Jamalipour, “Sensing with unknown signals: ISAC enabled distributed passive sensing for multi-target detection and localization," \textit{IEEE Trans. Netw. Sci. Eng.}, vol. 13, pp. 7599--7613, 2026.

\bibitem{sui2025ris}
Z. Sui, H. Q. Ngo, T. V. Chien, M. Matthaiou, and L. Hanzo, “RIS-assisted cell-free massive MIMO relying on reflection pattern modulation," \textit{IEEE Trans. Commun.}, vol. 73, no. 2, pp. 968--982, Feb. 2025

\bibitem{ahmed2026survey}
M. Ahmed et al., “Towards 6G networks: A survey on integrated sensing and communication in cell-free massive MIMO,” \textit{IEEE Internet Things J.}, doi: 10.1109/JIOT.2026.3693228.

\bibitem{galappaththige2025cell}
D. Galappaththige, M. Mohammadi, G. A. A. Baduge, and C. Tellambura, “Cell-free integrated sensing and communication: Principles, advances, and future directions," \textit{Proceedings IEEE}, vol. 113, no. 12, pp. 1418--1454, Dec. 2025.

\bibitem{behdad2024multistatic}
Z. Behdad, Ö. T. Demir, K. W. Sung, E. Björnson, and C. Cavdar, “Multi-static target detection and power allocation for integrated sensing and communication in cell-Free massive MIMO," \textit{IEEE Trans. Wireless Commun.}, vol. 23, no. 9, pp. 11580--11596, Sep. 2024.

\bibitem{demirhan2024cellfree}
U. Demirhan and A. Alkhateeb, ``Cell-free ISAC MIMO systems: Joint sensing and communication beamforming,'' \textit{IEEE Trans. Commun.}, vol. 73, no. 6, pp. 4454--4468, 2024.

\bibitem{guo2026multiap}
J. Guo, L. Li, Y. Zheng, D. Zhao, W. Lin, and Z. Han, “Multi-AP cooperative beamforming for cell-free ISAC networks: Balancing communication SINR and sensing SCNR," \textit{IEEE Wireless Commun. Lett.}, vol. 15, pp. 3004--3008, 2026.

\bibitem{salem2025integrated}
A. Abdelaziz Salem, M. A. Albreem, K. A. Alnajjar, S. Abdallah, and M. Saad, “Integrated cooperative sensing and communication for RIS-enabled full-duplex cell-free MIMO systems," \textit{IEEE Trans. Commun.}, vol. 73, no. 6, pp. 3804--3819, Jun. 2025.

\bibitem{wang2026cooperative}
Z. Wang, V. W. S. Wong, and R. Schober, “Cooperative ISAC for joint localization and velocity estimation in cell-free MIMO systems," \textit{IEEE J. Sel. Areas Commun.}, vol. 44, pp. 642--658, 2026.

\bibitem{liu2026cooperative}
H. Liu, Z. Wei, L. Sun, R. Xu, Y. Zhang, and Z. Feng, “Cooperative sensing in cell-free massive MIMO ISAC systems: Performance optimization and signal processing," \textit{IEEE Trans. Wireless Commun.}, vol. 25, pp. 12531--12547, 2026.

\bibitem{liu2020joint_mimo_radar}
X. Liu, T. Huang, N. Shlezinger, Y. Liu, J. Zhou, and Y. C. Eldar, ``Joint transmit beamforming for multiuser MIMO communications and MIMO radar,'' \textit{IEEE Trans. Signal Process.}, vol. 68, pp. 3929--3944, 2020.

\bibitem{zhao2022joint}
N. Zhao, Y. Wang, Z. Zhang, Q. Chang, and Y. Shen, ``Joint transmit and receive beamforming design for integrated sensing and communication,'' \textit{IEEE Commun. Lett.}, vol. 26, no. 3, pp. 662--666, Mar. 2022.

\bibitem{yang2024coordinated}
X. Yang, Z. Wei, J. Xu, Y. Fang, H. Wu, and Z. Feng, ``Coordinated transmit beamforming for networked ISAC with imperfect CSI and time synchronization,'' \textit{IEEE Trans. Wireless Commun.}, vol. 23, no. 12, pp. 18019--18035, Dec. 2024.

\bibitem{ren2023robust}
Z. Ren, L. Qiu, J. Xu, and D. W. K. Ng, ``Robust transmit beamforming for secure integrated sensing and communication,'' \textit{IEEE Trans. Commun.}, vol. 71, no. 9, pp. 5549--5564, Sep. 2023.

\bibitem{xu2013time}
K. Xu, Y. Xu, W. Ma, W. Xie, and D. Zhang, ``Time and frequency synchronization for multicarrier transmission on hexagonal time-frequency lattice,'' \textit{IEEE Trans. Signal Process.}, vol. 61, no. 24, pp. 6204--6219, Dec. 2013.

\bibitem{mostofi2007robust}
Y. Mostofi and D. C. Cox, ``A robust timing synchronization design in OFDM systems---Part I: Low-mobility cases,'' \textit{IEEE Trans. Wireless Commun.}, vol. 6, no. 12, pp. 4329--4339, Dec. 2007.

\bibitem{xiu2026robust}
Y. Xiu, Y. Zhao, R. Yang, M. Li, and R. Zhang, ``Robust optimization for movable antenna-aided cell-free ISAC with time synchronization errors,'' \textit{IEEE Trans. Wireless Commun.}, vol. 25, pp. 10082--10097, 2026.

\bibitem{principles}
M. A. Richards, J. A. Scheer, and W. A. Holm, Eds., \textit{Principles of Modern Radar: Basic Principles}. Raleigh, NC, USA: SciTech, 2010.

\bibitem{zhang2019precoding}
Y. Zhang, M. Xiao, S. Han, M. Skoglund, and W. Meng, ``On precoding and energy efficiency of full-duplex millimeter-wave relays,'' \textit{IEEE Trans. Wireless Commun.}, vol. 18, no. 3, pp. 1943--1956, Mar. 2019.

\end{thebibliography}

\end{document}